\documentclass{article}
\usepackage[utf8]{inputenc}
\usepackage{graphicx} 
\usepackage{fullpage}
\usepackage{amsmath}
\usepackage{amsthm}
\usepackage{amssymb}
\usepackage{enumerate}
\usepackage{xcolor}
\usepackage[sort,numbers]{natbib}
\usepackage{algorithm}
\usepackage{algorithmic}
\usepackage{comment}
\usepackage{enumitem}
\usepackage{hyperref}
\usepackage[normalem]{ulem}
\usepackage{authblk}

\usepackage{tikz-cd}
\newlength{\ptsize}
\newtheorem{lemma}{Lemma}

\newtheorem{theorem}{Theorem}
\newtheorem{cor}{Corollary}

\newtheorem{problem}{Open Problem}

\newcommand{\rSPR}{\mathrm{rSPR}}
\newcommand{\HOP}{\mathrm{HOP}}
\newcommand{\OLA}{\mathrm{OLA}}
\newcommand{\PV}{\mathrm{P2V}}
\newcommand{\ul}{\underline}

\usepackage{enumerate}

\title{Order-Dependent Dissimilarity Measures on Phylogenetic Trees}
\author[1]{Simone Linz}
\author[2,3]{Katherine St.~John}
\author[4]{Charles Semple}
\author[5]{Kristina Wicke}
\affil[1]{School of Computer Science, University of Auckland, New Zealand}
\affil[2]{Department of Computer Science, Hunter College, City University of New York, USA}
\affil[3]{Division of Invertebrate Zoology, American Museum of Natural History, USA}
\affil[4]{School of Mathematics and Statistics, University of Canterbury, New Zealand}
\affil[5]{Department of Mathematical Sciences, New Jersey Institute of Technology, USA}

\begin{document}

\maketitle

\begin{abstract}
Ordered leaf attachment, Phylo2Vec, and HOP are three recently introduced vector representations for rooted phylogenetic trees where the representation is determined by an ordering of the underlying leaf set $X$. Comparing the vectors of two rooted phylogenetic $X$-trees $T$ and $T'$ for a fixed ordering on $X$ leads to polynomial-time computable measure for the dissimilarity of $T$ and $T'$, albeit dependent on the choice of the leaf ordering.  For each of ordered leaf attachment, Phylo2Vec, and HOP, we compare this measure with the rooted subtree prune and regraft distance (rSPR), the hybrid number, and the temporal tree-child hybrid number of $T$ and $T'$. Although there is no direct relationship between rSPR and any of the three vector-based measures, we show that, when minimized over all orderings, the hybrid number is equivalent to HOP, and an upper bound on the other two. Moreover, when minimized over all orderings induced by common cherry-picking sequences of $T$ and $T'$, the temporal tree-child hybrid number of $T$ and $T'$ is equivalent to each of the three vector-based measures. 
\end{abstract}


\section{Introduction}\label{sec:intro}
The task of quantifying the disagreement between phylogenetic trees is essential for evaluating the accuracy of tree inference methods and for comparing phylogenetic trees \cite{kuhner2015practical,stjohn2017}.
One of the most common ways of doing this for two rooted binary phylogenetic trees is by using the tree rearrangement operation of rooted subtree prune and regraft (rSPR)~\cite{sempleSteel}. 
The rSPR distance between two arbitrary rooted binary phylogenetic $X$-trees $T$ and $T'$ is the minimum number of rSPR operations that transforms $T$ into $T'$. 
However, in general, computing the rSPR distance between $T$ and $T'$ is computationally hard~\cite{bordewichSemple2005}. 
To circumvent this computational hardness but to also replicate the applicability of rSPR, three new approaches have been recently introduced for this task. All three approaches, namely, ordered leaf attachment (OLA)~\cite{ola}, Phylo2Vec (P2V)~\cite{penn2024phylo2vec}, and HOP~\cite{hop}, are based on imposing an external ordering $\sigma$ on $X$ and, depending on $T$ and $\sigma$, using this ordering to assign a vector to $T$. For a particular measure, the distance between $T$ and $T'$ under $\sigma$ is made by a comparison of the vectors assigned to $T$ and $T'$. Regardless of the measure, the assignment of a vector to $T$ and the comparison of the vectors assigned to $T$ and $T'$ can be computed in polynomial time, thereby eliminating the hardness of calculating the rSPR distance.

The purpose of this paper is to investigate the relationship between each of the three order-dependent measures (OLA, P2V, and HOP) and the tree rearrangement operation rSPR that they seek to replicate. If we are allowed to choose an ordering on $X$, then there is little relationship between any of the three measures and rSPR. However, if we choose an ordering that minimizes the measure and broaden the investigation to include the hybridization number~\cite{baroni2005bounding}, a notion closely related to rSPR, then more direct relationships exist for OLA and HOP. In particular, across all orderings of $X$, we show that the minimum OLA measure of two rooted binary phylogenetic $X$-trees $T$ and $T'$ is bounded above by a linear function in the rSPR distance between $T$ and $T'$, while the minimum HOP measure of $T$ and $T'$ equates (exactly) to the hybrid number of $T$ and $T'$. The latter resolves a conjecture in \cite{hop} that computing the minimum HOP measure between two rooted binary phylogenetic trees is NP-hard as computing the hybrid number of two such trees is NP-hard \cite{bordewich2007computing}. Furthermore, by only considering the orderings of $X$ that are induced by a common cherry-picking sequence of $T$ and $T'$ and minimizing across these orderings, we show that the minimum OLA, P2V, and HOP measures all equate to the temporal tree-child hybrid number of $T$ and $T'$~\cite{humphries2013cherry}.

The above results require a number of concepts that are localized in their use, and so we delay their introduction until the relevant sections of the paper. The paper is organized as follows. In the next section, we formally define each of OLA, P2V, and HOP. Section~\ref{sec:results} establishes the above relationships between OLA and rSPR, and between HOP and the hybrid number, while Section~\ref{results:cherry-picking} establishes the relationship between all three order-dependent measures and the temporal tree-child hybrid number. The relationships in Section~\ref{sec:results} make use of agreement forests, while the relationship in Section~\ref{results:cherry-picking} involves cherry-picking sequences. We end the paper with a discussion and some open problems.

\section{Order-dependent measures}

Throughout the paper, $X$ denotes a non-empty finite set with $|X|=n$. 
We begin with some concepts on phylogenetic trees, and orderings and vectors, and then turn to defining the three order-dependent measures OLA, P2V, and HOP.

\noindent
{\bf Phylogenetic trees.}
A {\em rooted binary phylogenetic $X$-tree} $T$ is a rooted tree that satisfies the following three properties: (i) the unique root has in-degree zero and out-degree one, (ii) the leaves are bijectively labeled with the elements in $X$, and (iii) each remaining vertex has in-degree one and out-degree two. The set $X$ is the {\em label set} of $T$ and denoted by $L(T)$. If $|X|=1$, then $T$ contains a single edge that is incident with the root and the unique element in $X$. An example of a phylogenetic tree with $X=\{a, b, c, d, e\}$ is shown in Figure~\ref{fig:tree}. Let $T$ and $T'$ be two rooted binary phylogenetic $X$-trees with vertex set $V$ and $V'$, respectively. We say that $T$ and $T'$ are {\em isomorphic} if there exists a bijection $\phi: V \rightarrow V'$ with $\phi(x)=x$ for all $x\in X$ and $(u,v)$ is an edge in $T$ if and only if $(\phi(u),\phi(v))$ is an edge in $T'$ for all $u,v\in V$. If $T$ and $T'$ are isomorphic, we write $T\simeq T'$. Since all phylogenetic trees in this paper are rooted and binary, we refer to a rooted binary phylogenetic $X$-tree as simply a {\em phylogenetic tree} throughout the remainder of the paper.

Let $T$ be phylogenetic $X$-tree with root $\rho$, and let $Y \subseteq X \cup \{\rho \}$. We denote by $T(Y)$ the minimal rooted subtree of $T$ that connects the elements in $Y$. Furthermore, the {\em restriction of $T$ to $Y$}, denoted by $T|Y$, is the rooted tree that is obtained from $T(Y)$ by suppressing all vertices of in-degree one and out-degree one. We note that the definitions of the label set of a phylogenetic tree and of two phylogenetic trees being isomorphic naturally extend to restrictions of phylogenetic trees. If $T|Y$ and $T|((X\cup \{\rho\})-Y)$ are vertex disjoint, we call $T|Y$ a {\em pendant} subtree of $T$.
A pair of leaves $\{a,b\}$ of $T$ is a {\em cherry} if $a$ and $b$ have a common parent. More generally, two vertices $u$ and $v$ of $T$ are {\em siblings} if they have a common parent. 
Lastly, let $T$ be a phylogenetic $X$-tree and let $z$ be an element not in $X$. We say that $z$ has been {\em adjoined} to $T$ if we subdivide an edge of $T$ with a new vertex, $u$ say, and join $u$ and $z$ with a new edge $(u, z)$.

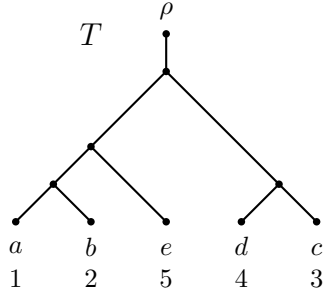
\begin{figure}[t!]
    \centering
    \begin{tikzpicture}[thick,scale=1]
            \node[outer sep=2.5pt, fill=black,circle,inner sep=1pt, label=below: {$a$} ] at (0.5,1){};
            \node[fill=black,circle,inner sep=1pt, label=above:$\rho$]  at (2.5,3.5){};
            \node[outer sep=1pt,fill=black,circle,inner sep=1pt, label=below: {$b$} ] at (1.5,1){};
            \node[outer sep=3pt,fill=black,circle,inner sep=1pt, label=below: {$c$} ] at (4.5,1){};
            \node[fill=black,circle,inner sep=1pt, label=below: {$d$} ] at (3.5,1){};
            \node[outer sep=3pt,fill=black,circle,inner sep=1pt, label=below: {$e$} ] at (2.5,1){};
            \node[fill=black,circle,inner sep=1pt] at (1,1.5){};
            \node[fill=black,circle,inner sep=1pt] at (2.5,3){};
            \node[fill=black,circle,inner sep=1pt] at (4,1.5){};
            \node[fill=black,circle,inner sep=1pt] at (1.5,2){};
            \node[align=left] at (0.5,0.25) {$1$};
            \node[align=left] at (1.5,0.25) {$2$};
            \node[align=left] at (2.5,0.25) {$5$};
            \node[align=left] at (3.5,0.25) {$4$};
            \node[align=left] at (4.5,0.25) {$3$};
            \draw(2.5,3.5)--(2.5,3);
            \draw(2.5,3)--(0.5,1);
            \draw(1,1.5)--(1.5,1);
            \draw(2.5,3)--(4.5,1);
            \draw(4,1.5)--(3.5,1);
            \draw(1.5,2)--(2.5,1);
            \node[align=left] at (1.5,3.5) {\large $T$};
             \end{tikzpicture}
    \caption{A rooted binary phylogenetic $X$-tree with root $\rho$ and $X = \{a,b,c,d,e\}$ together with an ordering $\sigma$ on $X$ defined by setting $\sigma(a)=1, \sigma(b)=2, \ldots, \sigma(e)=5$. The rank of each element $x \in X$ under $\sigma$ is given below the leaf label.}
    \label{fig:tree}
\end{figure}

\noindent
{\bf Orderings and vectors.}
An {\em ordering} on $X$ is a bijection $\sigma:X\rightarrow\{1, 2, \ldots, n\}$. For each $x \in X$, we refer to $\sigma(x)$ as the {\em rank} of $x$ under $\sigma$. To illustrate, consider the phylogenetic tree in Figure~\ref{fig:tree}. Here the ranking of $a$, $b$, $c$, $d$, and $e$ is $1$, $2$, $3$, $4$, and $5$, respectively. Let $\sigma$ be an ordering on $X$, and let $Y=\{y_1, y_2, \ldots, y_m\}$ be a subset of $X$. We say that the elements in $Y$ have {\em consecutive ranks} under $\sigma$ if
$$\{\sigma(y_1), \sigma(y_2), \ldots, \sigma(y_m)\}=\{i,i+1,\ldots,i+m-1\}$$
for some $i\leq n-m+1$. Furthermore, if $x$ is the element in $X$ such that $\sigma(x)=n$, we use $\sigma_{-x}$ to denote the ordering on $X-\{x\}$ such that $\sigma_{-x}(y)=\sigma(y)$ for each element $y\in X-\{x\}$. 

Let $\mathbf{v}$ be a vector. For the purposes of concatenating vectors, we write $[u, \mathbf{v}, w]$ for $[u,v_1, v_2, \ldots, v_k, w]$, where $\mathbf{v}= [v_1, v_2, \ldots, v_k]$. Now let $\mathbf{v}=[v_1, v_2,\ldots, v_n]$ be a vector of length $n$ and let $\sigma$ be an ordering on $X$. For each $i\in\{1,2,\ldots,n\}$, we say that $v_i$ is {\em associated} with the element $x\in X$ if $\sigma(x)=i$. Furthermore, for a subset $Y$ of $X$, the vector {\em $\mathbf{v}$ restricted to $Y$} is the vector with $|Y|$ coordinates that can be obtained from $\mathbf{v}$ by deleting each coordinate that is associated with an element in $X-Y$.

Lastly, we say that a vector $\mathbf{u}$ is a {\em subsequence} of a vector $\mathbf{v}$ if $\mathbf{u}$ can be obtained from $\mathbf{v}$ by the deletion of zero or more coordinates. More generally, let $S$ be a set of vectors. A vector is a {\em common subsequence} of $S$ if it is a subsequence of each vector in $S$. A vector is a {\em longest common subsequence (LCS)} of $S$ if it is (i) a subsequence of each vector in $S$ and (ii), among all such common subsequences, it has maximum length.

We next define the three order-dependent measures. In what follows, let $T$ be a phylogenetic $X$-tree with root $\rho$, and let $\sigma$ be an ordering on $X$. A {\em labeling} of $T$ will be a map $f:V(T) \rightarrow \mathcal{C}$, where $\mathcal{C}$ is a set of symbols. For each of the order-dependent measures, we present an algorithm whose input is a phylogenetic $X$-tree $T$ and an ordering $\sigma$ on $X$, and whose output is a vector associated with $T$ and $\sigma$. Each algorithm proceeds by first defining a labeling of $T$ or, in the case of P2V, a sequence of labelings of restrictions of $T$, and then using this labeling, respectively these labelings, to construct the outputted vector.

\subsection{Ordered leaf attachment (OLA)}
\label{sec:OLA}

OLA~\cite{ola} is the simplest of the three ordered-dependent measures to define. For OLA, the {\em OLA labeling} of $T$ under $\sigma$ is a map:
$$f_\OLA: (V(T) -\{\rho\})\rightarrow \{1, 2, \ldots, n\} \cup \{-2,-3,\ldots,-n\}.$$

\begin{algorithm}[H]
\caption{{\sc Construct OLA Vector}}
\begin{algorithmic}[1]
\STATE {\bf Input:} A phylogenetic $X$-tree with root $\rho$ and an ordering $\sigma$ on $X$.
\STATE {\bf Output:} The OLA vector $\mathbf{v}=[v_1, v_2, \ldots, v_n]$ of $T$ under $\sigma$.
\FORALL{$x\in X$}
\STATE Set $f_\OLA(x) = \sigma(x)$, i.e., each element of $X$ is labeled by its rank under $\sigma$.
\ENDFOR
\STATE Set $T_2\simeq T|\{\rho, \sigma^{-1}(1),\sigma^{-1}(2)\}$;
\STATE Set $u_2$ to be the unique child of $\rho$;
\STATE Set $f_\OLA(u_2)=-2$;
\FOR{$i = 3, 4, \ldots, n$}
\STATE Set $T_i$ to be the phylogenetic tree obtained from $T_{i-1}$ by adjoining $\sigma^{-1}(i)$ with the new edge $(u_i, \sigma^{-1}(i))$ so that $T_i\simeq T|\{\rho, \sigma^{-1}(1), \sigma^{-1}(2), \ldots, \sigma^{-1}(i)\}$;
\STATE Set $f_\OLA(u_i)=-i$;
\STATE Set $s_i$ to be the sibling of $\sigma^{-1}(i)$ in $T_i$;
\ENDFOR
\RETURN $\mathbf{v}=[0, 1, f_\OLA(s_3), f_\OLA(s_4), \ldots, f_\OLA(s_n)]$
\end{algorithmic}
\label{alg:ola}
\end{algorithm}

\noindent The vector returned by Algorithm~\ref{alg:ola} is the {\em OLA vector} of $T$ under $\sigma$ (see Figure~\ref{fig:OLA} for an illustration).

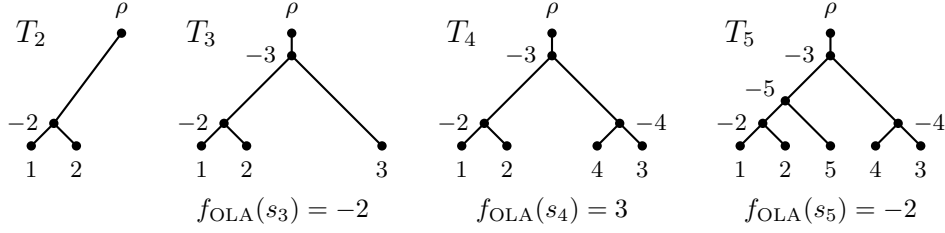
\begin{figure}[!t]
    \begin{center}
        \begin{tabular}{cccc}
             \begin{tikzpicture}[thick,scale=0.6]
            \node[fill=black,circle,inner sep=\ptsize,label=below: {\small $1$} ] at (0.5,1){};
            \node[fill=black,circle,inner sep=\ptsize,label=below: {\small $2$} ] at (1.5,1){};
            \node[fill=black,circle,inner sep=\ptsize,label=left: {\small $-2$} ] at (1,1.5){};
            \node[fill=black,circle,inner sep=\ptsize, label=above:{\small $\rho$}] at (2.5,3.5){};
            \draw(2.5,3.5)--(1,1.5);
            \draw(1,1.5)--(0.5,1);
            \draw(1,1.5)--(1.5,1);
             \node[align=left] at (0.5,3.5) {\large $T_2$};
             \end{tikzpicture}
             & 
             \begin{tikzpicture}[thick,scale=0.6]
            \node[fill=black,circle,inner sep=\ptsize,label=below: {\small $1$} ] at (0.5,1){};
            \node[fill=black,circle,inner sep=\ptsize,label=below: {\small $2$} ] at (1.5,1){};
            \node[fill=black,circle,inner sep=\ptsize,label=below: {\small $3$} ] at (4.5,1){};
            \node[fill=black,circle,inner sep=\ptsize,label=left: {\small $-2$} ] at (1,1.5){};
            \node[fill=black,circle,inner sep=\ptsize,label=left: {\small $-3$} ] at (2.5,3){};
            \node[fill=black,circle,inner sep=\ptsize, label=above:{\small $\rho$}] at (2.5,3.5){};
            \draw(2.5,3.5)--(2.5,3);
            \draw(2.5,3)--(0.5,1);
            \draw(1,1.5)--(1.5,1);
            \draw(2.5,3)--(4.5,1);
            \node[align=left] at (0.5,3.5) {\large $T_3$};
             \end{tikzpicture}
             &
             \begin{tikzpicture}[thick,scale=0.6]
            \node[fill=black,circle,inner sep=\ptsize,label=below: {\small $1$} ] at (0.5,1){};
            \node[fill=black,circle,inner sep=\ptsize,label=below: {\small $2$} ] at (1.5,1){};
            \node[fill=black,circle,inner sep=\ptsize,label=below: {\small $3$} ] at (4.5,1){};
            \node[fill=black,circle,inner sep=\ptsize,label=below: {\small $4$} ] at (3.5,1){};
            \node[fill=black,circle,inner sep=\ptsize,label=left: {\small $-2$} ] at (1,1.5){};
            \node[fill=black,circle,inner sep=\ptsize,label=left: {\small $-3$} ] at (2.5,3){};
            \node[fill=black,circle,inner sep=\ptsize,label=right: {\small $-4$} ] at (4,1.5){};
            \node[fill=black,circle,inner sep=\ptsize, label=above:{\small $\rho$}] at (2.5,3.5){};
            \draw(2.5,3.5)--(2.5,3);
            \draw(2.5,3)--(0.5,1);
            \draw(1,1.5)--(1.5,1);
            \draw(2.5,3)--(4.5,1);
            \draw(4,1.5)--(3.5,1);
            \node[align=left] at (0.5,3.5) {\large $T_4$};
             \end{tikzpicture}
             & 
             \begin{tikzpicture}[thick,scale=0.6]
            \node[fill=black,circle,inner sep=\ptsize,label=below: {\small $1$} ] at (0.5,1){};
            \node[fill=black,circle,inner sep=\ptsize,label=below: {\small $2$} ] at (1.5,1){};
            \node[fill=black,circle,inner sep=\ptsize,label=below: {\small $3$} ] at (4.5,1){};
            \node[fill=black,circle,inner sep=\ptsize,label=below: {\small $4$} ] at (3.5,1){};
            \node[fill=black,circle,inner sep=\ptsize,label=below: {\small $5$} ] at (2.5,1){};
            \node[fill=black,circle,inner sep=\ptsize,label=left: {\small $-2$} ] at (1,1.5){};
            \node[fill=black,circle,inner sep=\ptsize,label=left: {\small $-3$} ] at (2.5,3){};
            \node[fill=black,circle,inner sep=\ptsize,label=right: {\small $-4$} ] at (4,1.5){};
            \node[fill=black,circle,inner sep=\ptsize,label={[left, yshift=2]{\small $-5$}} ] at (1.5,2){};
            \node[fill=black,circle,inner sep=\ptsize, label=above:{\small $\rho$}] at (2.5,3.5){};
            \draw(2.5,3.5)--(2.5,3);
            \draw(2.5,3)--(0.5,1);
            \draw(1,1.5)--(1.5,1);
            \draw(2.5,3)--(4.5,1);
            \draw(4,1.5)--(3.5,1);
            \draw(1.5,2)--(2.5,1);
            \node[align=left] at (0.5,3.5) {\large $T_5$};
             \end{tikzpicture}
             \\
             & $f_\OLA(s_3) = -2$ & $f_\OLA(s_4) = 3$ & $f_\OLA(s_5)=-2$ \\
        \end{tabular}
    \end{center}
    \caption{Illustration of the OLA vector construction (Algorithm~\ref{alg:ola}) for the rooted phylogenetic $X$-tree on $X = \{a,b,c,d,e\}$ depicted in Figure~\ref{fig:tree} under the ordering $\sigma$ with $\sigma(a)=1, \sigma(b)=2, \ldots, \sigma(e)=5$. The labels of the interior vertices correspond to the OLA labeling, $f_\OLA$, and for $i \in \{3,4,5\}$, we indicate the $i$-th vector coordinate of the OLA vector $\mathbf{v}$ of $T$ below tree $T_i$. It follows that $\mathbf{v} = [0,1, f_\OLA(s_3), f_\OLA(s_4), f_\OLA(s_5) ] = [0,1,-2,3,-2]$.}
    \label{fig:OLA}
\end{figure}

\subsection{Phylo2Vec (P2V)} \label{sec:P2V}
Similar to OLA, the P2V labeling~\cite{penn2024phylo2vec,scheidwasser2025phylo2vec} iteratively assigns labels to the interior vertices of a phylogenetic tree based on an ordering of $X$. However, while OLA assigns a fixed negative integer to each interior vertex that remains unchanged as new leaves are adjoined, P2V recomputes the labels of these vertices after each new leaf is adjoined.

For all $i\in \{2, 3, \ldots, n\}$, the {\em $i$-th $\PV$ labeling} of $T_i\simeq T|\{\rho, \sigma^{-1}(1), \sigma^{-1}(2), \ldots, \sigma^{-1}(i)\}$ under $\sigma$ is a map 
$$f^i_\PV: (V(T_i)-\{\rho\}) \rightarrow \{1, 2, \ldots, i, i+1, \ldots, 2i-1\}.$$

\begin{algorithm}[H]
\caption{{\sc Construct P2V Vector}}
\begin{algorithmic}[1]
\STATE {\bf Input:} A phylogenetic $X$-tree $T$ with root $\rho$ and an ordering $\sigma$ on $X$.
\STATE {\bf Output:} The P2V vector $\mathbf{v}=[v_1, v_2, \ldots, v_n]$ of $T$ under $\sigma$.
\FORALL{$i\in \{2, 3, \ldots, n\}$ and $x\in X$}
\STATE Set $f^i_\PV(x) = \sigma(x)$, i.e., each element of $X$ is labeled by its rank under $\sigma$;
\ENDFOR
\STATE Set $T_2\simeq T|\{\rho,\sigma^{-1}(1),\sigma^{-1}(2)\}$;
\STATE Set $u_2$ to be the child of $\rho$;
\STATE Set $f^2_\PV(u_2)=3$;
\FOR{$i=3, 4, \ldots, n$}
\STATE Set $T_i$ to be the phylogenetic tree obtained from $T_{i-1}$ by adjoining $\sigma^{-1}(i)$ with the new edge $(u_i, \sigma^{-1}(i))$ so that $T_i\simeq T|\{\rho, \sigma^{-1}(1), \sigma^{-1}(2), \ldots, \sigma^{-1}(i)\}$;
\FOR{$j=i+1, i+2, \ldots, 2i-1$}
\STATE\label{alg:cherry} Assign $j$ to the vertex, $w_j$ say, in $\{u_2, u_3, \ldots, u_i\}$ that has no $i$-th P2V label but whose two children have an $i$-th P2V label and, amongst all such vertices, has the highest $i$-th P2V labeled child;
\STATE Set $f^i_\PV(w_j)=j$;
\ENDFOR
\STATE Set $s_i$ to be the sibling of $\sigma^{-1}(i)$ in $T_i$;
\ENDFOR
\RETURN $\mathbf{v}=[0, 1, f^3_\PV(s_3), f^4_\PV(s_4), \ldots, f^n_\PV(s_n)]$
\end{algorithmic}
\label{alg:p2v}
\end{algorithm}

\noindent The vector returned by Algorithm~\ref{alg:p2v} is the {\em P2V vector} of $T$ under $\sigma$ (see Figure~\ref{fig:P2V} for an illustration). Note that Step~\ref{alg:cherry} is well defined as every phylogenetic tree has a cherry and so there is always at least one interior vertex with both of its children labeled. Furthermore, noting that $T_n\simeq T$, we denote the $n$-th P2V labeling of $T_n$ by omitting the superscript and writing $f_\PV$.

\begin{figure}[t!]
    \begin{center}
        \begin{tabular}{cccc}
             \begin{tikzpicture}[thick,scale=0.6]
            \node[fill=black,circle,inner sep=\ptsize,label=below: {\small $1$} ] at (0.5,1){};
            \node[fill=black,circle,inner sep=\ptsize,label=below: {\small $2$} ] at (1.5,1){};
            \node[fill=black,circle,inner sep=\ptsize,label={[left, yshift=2]{\small $3$}} ] at (1,1.5){};
            \node[fill=black,circle,inner sep=\ptsize, label=above: {\small $\rho$}] at (2.5,3.5){};
            \draw(2.5,3.5)--(1,1.5);
            \draw(1,1.5)--(0.5,1);
            \draw(1,1.5)--(1.5,1);
             \node[align=left] at (0.5,3.5) {\large $T_2$};
             \end{tikzpicture}
             & 
             \begin{tikzpicture}[thick,scale=0.6]
            \node[fill=black,circle,inner sep=\ptsize,label=below: {\small $1$} ] at (0.5,1){};
            \node[fill=black,circle,inner sep=\ptsize,label=below: {\small $2$} ] at (1.5,1){};
            \node[fill=black,circle,inner sep=\ptsize,label=below: {\small $3$} ] at (4.5,1){};
            \node[fill=black,circle,inner sep=\ptsize,label={[left, yshift=2]{\small $4$}} ] at (1,1.5){};
           \node[fill=black,circle,inner sep=\ptsize,label={[left, yshift=0]{\small $5$}} ] at (2.5,3){};
             \node[fill=black,circle,inner sep=\ptsize, label=above: {\small $\rho$}] at (2.5,3.5){};
            \draw(2.5,3.5)--(2.5,3);
            \draw(2.5,3)--(0.5,1);
            \draw(1,1.5)--(1.5,1);
            \draw(2.5,3)--(4.5,1);
             \node[align=left] at (0.5,3.5) {\large $T_3$};
             \end{tikzpicture}
             &
             \begin{tikzpicture}[thick,scale=0.6]
            \node[fill=black,circle,inner sep=\ptsize,label=below: {\small $1$} ] at (0.5,1){};
            \node[fill=black,circle,inner sep=\ptsize,label=below: {\small $2$} ] at (1.5,1){};
            \node[fill=black,circle,inner sep=\ptsize,label=below: {\small $3$} ] at (4.5,1){};
            \node[fill=black,circle,inner sep=\ptsize,label=below: {\small $4$} ] at (3.5,1){};
            \node[fill=black,circle,inner sep=\ptsize,label={[left, yshift=2]{\small $6$}} ] at (1,1.5){};
             \node[fill=black,circle,inner sep=\ptsize, label=above: {\small $\rho$}] at (2.5,3.5){};
            \node[fill=black,circle,inner sep=\ptsize,label={[right, yshift=2]{\small $5$}} ]  at (4,1.5){};
            \node[fill=black,circle,inner sep=\ptsize,label={[left,yshift=2]{\small $7$}} ] at (2.5,3){};
            \draw(2.5,3.5)--(2.5,3);
            \draw(2.5,3)--(0.5,1);
            \draw(1,1.5)--(1.5,1);
            \draw(2.5,3)--(4.5,1);
            \draw(4,1.5)--(3.5,1);
            \node[align=left] at (0.5,3.5) {\large $T_4$};
             \end{tikzpicture}
             & 
             \begin{tikzpicture}[thick,scale=0.6]
            \node[fill=black,circle,inner sep=\ptsize,label=below: {\small $1$} ] at (0.5,1){};
            \node[fill=black,circle,inner sep=\ptsize,label=below: {\small $2$} ] at (1.5,1){};
            \node[fill=black,circle,inner sep=\ptsize,label=below: {\small $3$} ] at (4.5,1){};
            \node[fill=black,circle,inner sep=\ptsize,label=below: {\small $4$} ] at (3.5,1){};
            \node[fill=black,circle,inner sep=\ptsize,label=below: {\small $5$} ] at (2.5,1){};
           \node[fill=black,circle,inner sep=\ptsize,label={[left, yshift=2]{\small $7$}} ] at (1,1.5){};
            \node[fill=black,circle,inner sep=\ptsize,label={[left,yshift=2]{\small $9$}} ] at (2.5,3){};
            \node[fill=black,circle,inner sep=\ptsize,label=right: {\small $6$} ] at (4,1.5){};
            \node[fill=black,circle,inner sep=\ptsize,label={[left, yshift=2]{\small $8$}} ] at (1.5,2){};
             \node[fill=black,circle,inner sep=\ptsize, label=above: {\small $\rho$}] at (2.5,3.5){};
            \draw(2.5,3.5)--(2.5,3);
            \draw(2.5,3)--(0.5,1);
            \draw(1,1.5)--(1.5,1);
            \draw(2.5,3)--(4.5,1);
            \draw(4,1.5)--(3.5,1);
            \draw(1.5,2)--(2.5,1);
            \node[align=left] at (0.5,3.5) {\large $T_5$};
             \end{tikzpicture}
             \\
             & $f_\PV^3(s_3) = 4$ & $f_\PV^4(s_4) = 3$ & $f_\PV^5(s_5)=7$ \\
        \end{tabular}
    \end{center}
    \caption{Illustration of the P2V vector construction (Algorithm~\ref{alg:p2v}) for the rooted phylogenetic $X$-tree on $X = \{a,b,c,d,e\}$ depicted in Figure~\ref{fig:tree} under the ordering $\sigma$ with $\sigma(a)=1, \sigma(b)=2, \ldots, \sigma(e)=5$. For $i \in \{2,3,4,5\}$, the labels of the interior vertices of $T_i$ correspond to the P2V labeling $f^i_\PV$. Notice that, in contrast to OLA (Figure~\ref{fig:OLA}), the interior vertices are re-labeled in each step. Furthermore, for $i \in \{3,4,5\}$, we indicate the $i$-th vector coordinate of the P2V vector $\mathbf{v}$ of $T$
    below tree $T_i$. It follows that $\mathbf{v} = [0,1, f_\PV^3(s_3),f_\PV^4(s_4), f_\PV^5(s_5) ] = [0,1,4,3,7]$.}
    \label{fig:P2V}
\end{figure}
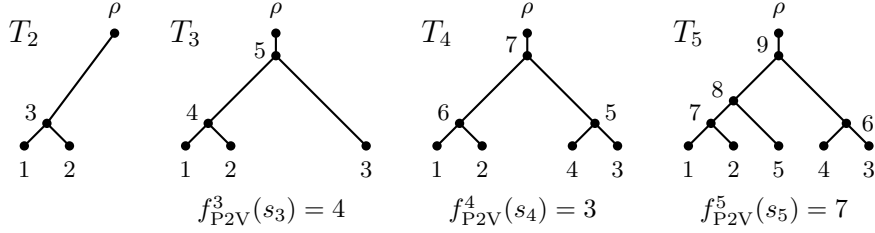

\subsection{HOP} 
Introduced in~\cite{hop}, the HOP labeling assigns a positive integer to each interior vertex of $T$ including the root such that there is a bijection between the interior vertices of $T$ and the elements in $\{1,2,\ldots,n\}$. The associated vector with $2n$ coordinates that arise from the $n$ leaves and the $n$ interior vertices of $T$ is then obtained by sequentially decomposing $T$ into $n$ edge-disjoint paths.

The {\em HOP labeling} of $T$ under $\sigma$ is a map
$$f_{\rm HOP}: V(T)\rightarrow \{\underline{1}, \underline{2}, \ldots, \underline{n}, 1, 2, \ldots, n\},$$
where in keeping with \cite{hop} the underlined numerals correspond to the leaves of $T$ and the numerals not underlined correspond to the interior vertices of $T$.

\begin{algorithm}[H]
    \caption{{\sc Construct HOP Vector}}
    \begin{algorithmic} [1]
       \STATE \textbf{Input:} A phylogenetic $X$-tree $T$ with root $\rho$ and an ordering $\sigma$ on $X$.
       \STATE \textbf{Output:} The HOP vector $\mathbf{v}=[v_1, v_2, \ldots, v_{2n}]$ of $T$ under $\sigma$.
       \FORALL{$x \in X$}
            \STATE Set $f_\HOP(x) = \underline{\sigma(x)}$;
        \ENDFOR
       \FORALL{$v \in V(T)$}
            \STATE Set $C(v)$ to be the set of leaves $x\in X$ with a directed path from $v$ to $x$ in $T$;
            \STATE  Compute $m(v) = \min_{x \in C(v)} f_\HOP(x)$, i.e., find the minimum HOP label among the leaves in $C(v)$;
       \ENDFOR
       \FORALL{$v \in V(T) -( L(T)\cup\{\rho\})$}
            \STATE Label $v$ by $f_\HOP(v) = \max\{ m(v_1), m(v_2)\}$, where $v_1$ and $v_2$ are the children of $v$; 
        \ENDFOR
        \STATE Set $f_\HOP(\rho)=1$;
    \FOR{$i=1,2, \ldots, n$}
        \STATE Let $x$ be the leaf with $\sigma(i)=x$ and $v$ be the non-leaf vertex with $f_\HOP(v)=i$;
        \STATE Let 
        $P_i=(v=u_1, u_2, \ldots, u_{k-1}, u_k=x)$ be the unique path from $v$ to $x$ in $T$;
        \STATE Set $\mathbf{v}(P_i) = [ f_\HOP(u_2),f_\HOP(u_3), \ldots, f_\HOP(u_{k-1})]$;
    \ENDFOR
    \RETURN $\mathbf{v} = [1, \mathbf{v}(P_1), \ul{1}, \mathbf{v}(P_2), \ul{2}, \mathbf{v}(P_3), \ul{3}, \ldots, \mathbf{v}(P_n),\ul{n}]$ \label{abbreviation}
    \end{algorithmic}
    \label{alg:hop}
\end{algorithm}

\noindent The vector returned by Algorithm~\ref{alg:hop} is the {\em HOP vector} of $T$ under $\sigma$. We sometimes abbreviate $\mathbf{v}(P_i)$ in Line~17 of Algorithm~\ref{alg:hop} as $\mathbf{v}_i$. Note that, by definition, $\mathbf{v}(P_n) = \mathbf{v}_n$ is empty, and so this leads to $\mathbf{v} = [1, \mathbf{v}_1, \underline{1}, \mathbf{v}_2, \underline{2}, \mathbf{v}_3, \underline{3}, \ldots, \mathbf{v}_{n-1}, \underline{n-1},\underline{n}]$.

We remark here that the HOP labeling of $T$ under $\sigma$ can alternatively be derived through a tree-growing process similar to the one described for OLA. As this derivation is not needed for the paper, we omit the details.

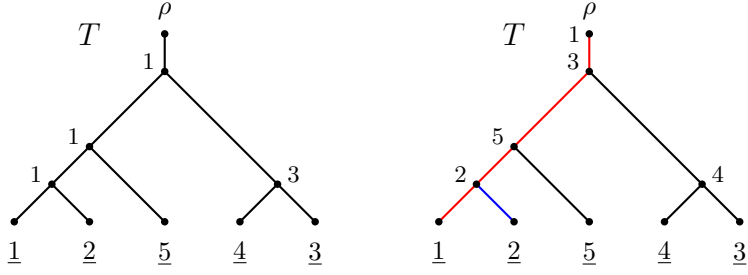
\begin{figure}[!t]
\begin{center}
\begin{tikzpicture}[thick,scale=1]
    \node[outer sep=2.5pt, fill=black,circle,inner sep=1pt, label=below: {$\ul{1}$} ] at (0.5,1){};
    \node[fill=black,circle,inner sep=1pt, label=above:$\rho$]  at (2.5,3.5){};
    \node[outer sep=2.5pt,fill=black,circle,inner sep=1pt, label=below: {$\ul{2}$} ] at (1.5,1){};
    \node[outer sep=3pt,fill=black,circle,inner sep=1pt, label=below: {$\ul{3}$} ] at (4.5,1){};
    \node[outer sep=2.5pt,fill=black,circle,inner sep=1pt, label=below: {$\ul{4}$} ] at (3.5,1){};
    \node[outer sep=3pt,fill=black,circle,inner sep=1pt, label=below: {$\ul{5}$} ] at (2.5,1){};
    \node[fill=black,circle,inner sep=1pt, label={[left, yshift=2]{\small $1$}} ] at (1,1.5){};
    \node[fill=black,circle,inner sep=1pt, label={[left, yshift=2]{\small $1$}} ] at (2.5,3){};
    \node[fill=black,circle,inner sep=1pt, label={[right, yshift=2]{\small $3$}} ] at (4,1.5){};
    \node[fill=black,circle,inner sep=1pt, label={[left, yshift=2]{\small $1$}} ] at (1.5,2){};
    \draw(2.5,3.5)--(2.5,3);
    \draw(2.5,3)--(0.5,1);
    \draw(1,1.5)--(1.5,1);
    \draw(2.5,3)--(4.5,1);
    \draw(4,1.5)--(3.5,1);
    \draw(1.5,2)--(2.5,1);
    \node[align=left] at (1.5,3.5) {\large $T$};
\end{tikzpicture}
\hspace{1cm}
\begin{tikzpicture}[thick,scale=1]
    \draw[red](2.5,3.5)--(2.5,3);
    \draw[red](2.5,3)--(0.5,1);
    \draw[blue](1,1.5)--(1.5,1);
    \node[outer sep=2.5pt, fill=black,circle,inner sep=1pt, label=below: {$\ul{1}$} ] at (0.5,1){};
    \node[fill=black,circle,inner sep=1pt, label=above:$\rho$]  at (2.5,3.5){};
    \node[align=left] at (2.3,3.5) {\small $1$};
    \node[outer sep=2.5pt,fill=black,circle,inner sep=1pt, label=below: {$\ul{2}$} ] at (1.5,1){};
    \node[outer sep=3pt,fill=black,circle,inner sep=1pt, label=below: {$\ul{3}$} ] at (4.5,1){};
    \node[outer sep=2.5pt,fill=black,circle,inner sep=1pt, label=below: {$\ul{4}$} ] at (3.5,1){};
    \node[outer sep=3pt,fill=black,circle,inner sep=1pt, label=below: {$\ul{5}$} ] at (2.5,1){};
    \node[fill=black,circle,inner sep=1pt, label={[left, yshift=2]{\small $2$}} ] at (1,1.5){};
    \node[fill=black,circle,inner sep=1pt, label={[left, yshift=2]{\small $3$}} ] at (2.5,3){};
    \node[fill=black,circle,inner sep=1pt, label={[right, yshift=2]{\small $4$}} ] at (4,1.5){};
    \node[fill=black,circle,inner sep=1pt, label={[left, yshift=2]{\small $5$}} ] at (1.5,2){};
    \draw(2.5,3)--(4.5,1);
    \draw(4,1.5)--(3.5,1);
    \draw(1.5,2)--(2.5,1);
    \node[align=left] at (1.5,3.5) {\large $T$};
\end{tikzpicture}\\
\end{center}
\caption{Illustration of the HOP vector construction (Algorithm~\ref{alg:hop}) for the rooted phylogenetic $X$-tree on $X = \{a,b,c,d,e\}$ depicted in Figure~\ref{fig:tree} under the ordering $\sigma$ with $\sigma(a)=1$, $\sigma(b)=2, \ldots, \sigma(e)=5$. On the left, each  vertex $v \in V(T) - \{\rho\}$ is labeled by $m(v)$, the minimum HOP label among the leaves below $v$, with leaf labels underlined. On the right, the root is labeled by $1$, the leaves are labeled by their ranks (underlined), and each vertex $v \in V(T) - (L(T) \cup \{\rho\})$ with children $v_1, v_2$ is labeled by $\max \{m(v_1), m(v_2)\}$. For further illustration, the red edges correspond to the path $P_1$, and similarly, the blue edge corresponds to the path $P_2$. Finally, the HOP vector of $T$ is given by $\mathbf{v}=[1,3,5,2,\ul{1},\ul{2},4,\ul{3},\ul{4},\ul{5}]$.}
\end{figure}

\subsection{Order-dependent distances between phylogenetic trees}
The primary motivation for the introduction of the OLA, P2V, and HOP vector representations is to quantify the dissimilarity between two phylogenetic $X$-trees $T$ and $T'$. This dissimilarity depends on the choice of $\sigma$, an ordering of $X$, and is computed as follows:
\begin{enumerate}
\item The {\em OLA distance with respect to $\sigma$} between $T$ and $T'$, denoted as $d^\sigma_{\OLA}(T,T')$, is defined as the Hamming distance between the OLA vectors of $T$ and $T'$ under $\sigma$.

\item Similarly, the {\em $\PV$ distance with respect to $\sigma$} between $T$ and $T'$, denoted as $d^\sigma_{\PV}(T,T')$, is defined as the Hamming distance between the P2V vectors of $T$ and $T'$ under $\sigma$.  
    
\item The analogue for HOP is a two-step process. Let $\mathbf{u}=[1,\mathbf{u}_1, \underline{1}, \mathbf{u}_2, \underline{2}, \ldots, \mathbf{u}_{n-1},\underline{n-1},\underline{n}]$ and $\mathbf{v}=[1,\mathbf{v}_1, \underline{1}, \mathbf{v}_2, \underline{2}, \ldots, \mathbf{v}_{n-1},\underline{n-1},\underline{n}]$ be the HOP vectors of $T$ and $T'$ under $\sigma$. The {\em HOP similarity with respect to $\sigma$} is defined as 
    \[ \text{Sim}^\sigma_\HOP(T,T') = \sum\limits_{1 \leq i \leq n-1} \left| \text{LCS}(\mathbf{u}_i,\mathbf{v}_i) \right|.\] 
In turn, the {\em HOP distance with respect to $\sigma$} between $T$ and $T'$ is given by 
    \[ d^\sigma_\HOP(T,T') = n - \text{Sim}^\sigma_\HOP(\textbf{u},\textbf{v}) = n -  \sum\limits_{1 \leq i \leq n-1} \left| \text{LCS}(\mathbf{u}_i,\mathbf{v}_i) \right|.\]
\end{enumerate}
It has been shown in \cite{ola}, \cite{penn2024phylo2vec}, and \cite{hop}, respectively, that under a fixed $\sigma$ each of $d_\OLA^\sigma(T,T')$, $d_\PV^\sigma(T,T')$, and $d_\HOP^\sigma(T,T')$ satisfies the triangle inequality and, therefore, is a distance on the set of phylogenetic $X$-trees. 

Lastly, for each $\Theta\in \{\HOP, \OLA, \PV\}$, we define
$$d^*_{\Theta}(T, T')= \min\big\{d^{\sigma}_{\Theta}(T, T'): \mbox{$\sigma$ is an ordering on $X$}\big\},$$
and refer to $d^*_{\Theta}(T, T')$ as the {\it $\Theta$ measure between} $T$ and $T'$. 

The reason for calling it a measure  is because $d^*_{\Theta}(T, T')$ is not a distance on the set of phylogenetic $X$-trees. We give concrete examples showing that $d^*_{\Theta}(T, T')$ does not satisfy the triangle inequality in Figure~\ref{fig:OLA-rSPR} (for OLA) and Figure~\ref{fig:P2V-SPR} (for HOP and P2V).

\section{Bounding order-dependent measures by agreement forests} \label{sec:results}

There are several measures to compute the dissimilarities between two phylogenetic trees on the same label set that are based on agreement forests. These measures include the rooted subtree prune and regraft distance  and the hybrid number that we formally define next.

Let $T$ be a rooted phylogenetic $X$-tree with root $\rho$. Furthermore, let $T'$ be a rooted phylogenetic tree that can be obtained from $T$ by deleting an edge $(u,v)$ in $T$ with $u\ne\rho$, suppressing $u$, and then {\em adjoining} the subtree with root $v$ to the phylogenetic tree that contains $\rho$ by subdividing an edge of the latter with a new vertex, $u'$ say, and then adding the edge $(u',v)$. We say that $T'$ has been
obtained from $T$ by a {\em rooted subtree prune and regraft (rSPR) move}.  
Moreover, we define the {\em rSPR distance}, denoted $d_\rSPR(T,T')$, to be the minimum number of rSPR moves that transforms $T$ to $T'$. Note that an rSPR move is reversible, so $d_\rSPR(T, T')=d_\rSPR(T', T)$.

Now, let $T$ and $T'$ be two phylogenetic $X$-trees whose roots are labeled with $\rho$. For the purpose of the upcoming definition, we view $\rho$ as an element of the label set of $T$ and $T'$.
An {\em agreement forest} for $T$ and $T'$ is a collection $\{T_{\rho}, T_1, T_2,\ldots,T_k\}$ of rooted trees with label sets $L_{\rho}, L_1, L_2,\ldots, L_k$ such that the following properties are satisfied:
\begin{enumerate}
    \itemsep 0pt
    \item The label sets $L_{\rho}, L_1, L_2,\ldots, L_k$ partition $X \cup \{\rho\}$ and, in particular, $\rho\in L_\rho$. 
    \item For all $i\in \{\rho,1,2,\ldots,k\}$, the rooted tree $T_i \simeq T|{L_i} \simeq T'|{L_i}$. 
    \item The trees in $\{T(L_i): i\in \{\rho,1,2,\ldots,k\}\}$ and $\{T'(L_i): i\in \{\rho,1,2,\ldots,k\}\}$ are vertex-disjoint rooted subtrees of $T$ and $T'$, respectively.
\end{enumerate}
A {\em maximum agreement forest} for $T$ and $T'$ is an agreement forest $\{T_{\rho}, T_1, T_2,\ldots,T_k\}$ in which $k$ (the number of components minus one) is minimized. The minimum possible value for $k$ is denoted by $m(T,T')$. The next theorem is due to Bordewich and Semple~\cite[Theorem 2.1]{bordewichSemple2005} and establishes a characterization of the rSPR distance between two phylogenetic $X$-trees in terms of agreement forests.

\begin{theorem}\label{t:spr}
Let $T$ and $T'$ be two phylogenetic $X$-trees. Then $d_\rSPR(T,T')=m(T,T')$.
\end{theorem}

\noindent Following on from Theorem~\ref{t:spr}, Bordewich and Semple~\cite{bordewichSemple2005} showed that computing the rSPR distance between two phylogenetic trees is NP-hard via a reduction from Exact Cover by $3$-Sets to maximum agreement forests.  \medskip

In order to define the hybrid number of two phylogenetic $X$-trees $T$ and $T'$, we first introduce phylogenetic networks. A {\em rooted binary phylogenetic network $N$ on $X$} is a rooted acyclic directed graph with no parallel edges that satisfies the following properties:
\begin{enumerate}
    \itemsep 0pt
    \item the unique root has in-degree zero and  out-degree one,
    \item vertices with out-degree zero have in-degree one, and the set of vertices with out-degree zero is $X$, and
    \item all other vertices have either in-degree one and out-degree two or in-degree two and out-degree one.
\end{enumerate}
Furthermore, we use $h(N)$ to denote the number of vertices with in-degree two in $N$. Since all phylogenetic networks in this paper are rooted and binary, we refer to a rooted binary phylogenetic network on $X$ as simply a {\em phylogenetic network} from now on.

We next define two particular classes of phylogenetic networks. The first features prominently in the literature. Let $N$ be a phylogenetic network on $X$. We say that $N$ is {\em tree-child} if each non-leaf vertex has a child with in-degree one. Moreover, we say that $N$ is {\em temporal} if there exists a map $t: V\rightarrow {\mathbb R}^+$ such that, for each edge $(u,v)$ in $N$, we have $t(u)=t(v)$ if $v$ has in-degree two, and $t(u) < t(v)$ if $v$ has in-degree one. 

Now, let $N$ be a phylogenetic network on $X$, and let $T$ and $T'$ be two phylogenetic $X$-trees. We say that $N$ {\em displays} $T$ if there exists a subtree of $N$ that is a subdivision of $T$. 

With this definition in hand, we set
$$h(T,T')=\min\{h(N): N \text{ is a phylogenetic network on } X \text{ that displays }T \text{ and }T'\}$$
and refer to $h(T,T')$ as the {\em hybrid number} of $T$ and $T'$. Similarly, we set 
$$h_{tc}(T,T')=\min\{h(N): N \text{ is a tree-child network on } X \text{ that displays }T \text{ and }T'\}$$
and
$$h_t(T,T')=\min\{h(N): N \text{ is a temporal tree-child network on } X \text{ that displays }T \text{ and }T'\},$$
and refer to $h_{tc}(T,T')$ and $h_t(T,T')$ as the {\em tree-child hybrid number} and the {\em temporal tree-child hybrid number}, respectively, of $T$ and $T'$. The following lemma was essentially established as part of the proof of~\cite[Theorem 2]{baroni2005bounding} and formally noted  in~\cite{humphries2013cherry}. It shows that the hybrid number of two phylogenetic trees is equal to their tree-child hybrid number.

\begin{lemma}\label{l:tc-hybrid}
Let $T$ and $T'$ be two phylogenetic $X$-trees. Then $h(T,T')=h_{tc}(T,T')$.
\end{lemma}

We next define a concept that is closely related to agreement forests. Let $T$ and $T'$ be two phylogenetic $X$-trees, and let $F=\{T_\rho,T_1,T_2,\ldots,T_k\}$ be an agreement forest for $T$ and $T'$. We say that $F$ is {\em acyclic} if the directed graph with vertex set $F$ and for which $(T_i,T_j)$ with $i,j\in\{\rho,1,2,\ldots,k\}$ and $i\ne j$ is an edge precisely if the root of $T(L(T_i))$ is an ancestor of the root of $T(L(T_j))$, or the root of $T'(L(T_i))$ is an ancestor of the root of $T'(L(T_j))$ has no directed cycles. Note that such an agreement forest always exists. Now, let $F=\{T_\rho,T_1,T_2,\ldots,T_k\}$ be an acyclic agreement forest for $T$ and $T'$. Similar to agreement forests, we say that $F$ is a {\em maximum acyclic agreement forest} for $T$ and $T'$ if $k$ is minimum over all acyclic agreement forests for $T$ and $T'$. We denote this minimum number by $m_a(T,T')$. The relevance of acyclic agreement forests is the next theorem established in~\cite[Theorem 2]{baroni2005bounding}.

\begin{theorem} \label{thm:hybridnumber}
Let $T$ and $T'$ be two  phylogenetic $X$-trees. Then $h(T,T')=m_a(T,T')$.
\end{theorem}

\noindent Similar to computing the rSPR distance, Bordewich and Semple~\cite{bordewich2007computing} showed that computing the hybrid number for 
two phylogenetic trees is NP-hard. Moreover, by Lemma~\ref{l:tc-hybrid}, it immediately follows that computing the tree-child hybrid number is also NP-hard.

\noindent{\bf Notational remark.} In what follows, we sometimes compare the $\Theta$ labelings of two phylogenetic trees where  $\Theta\in\{\HOP,\OLA,\PV\}$.  In this case, we write $f_\Theta ^T$ instead of $f_\Theta$ to make a direct reference to a phylogenetic tree $T$ with $n$ leaves. For $\Theta=\PV$, we write $f_\PV^T$ to refer to the $\PV$ labeling of $T_n\simeq T$.

\subsection{Bounding by the rSPR distance}
\label{results:rspr}
In this section, we focus on the relationship between the order-dependent measures OLA, P2V, and HOP, and the rSPR distance. 
We begin by showing that given two phylogenetic $X$-trees $T$ and $T'$, the measure $d^\ast_\OLA(T,T')$ is bounded from above by a function that is linear in the rSPR distance $d_\rSPR(T,T')$ of $T$ and $T'$.

\begin{theorem} \label{thm:OLA-20SPR}
    Let $T$ and $T'$ be two phylogenetic $X$-trees.  Then, there exists an ordering $\sigma$ such that 
    \[
    d^{\sigma}_\OLA(T,T') \leq 28\cdot d_\rSPR(T,T').
    \]
    In particular, $d^\ast_\OLA(T,T')  \leq 28\cdot d_\rSPR(T,T')$.
\end{theorem}

In order to prove this statement, we first show that two tree reduction rules can be used to reduce the size of the label set of a pair of phylogenetic $X$-trees while preserving the OLA distance with respect to $\sigma$ between them. These two rules coincide with the two tree reduction rules used by~\cite{bordewichSemple2005} for the rSPR distance. \\

Let $T$ be a phylogenetic $X$-tree, and let $C=(x_1,x_2,\ldots,x_l)$ be a sequence of elements in $X$ with $l\geq 3$.  We say that $C$ is a {\em chain} of $T$ if the parent of $x_1$ coincides with the parent of $x_2$ or the parent of $x_2$ is the parent of the parent of $x_1$, and, for each $i\in\{3,4,\ldots,l\}$, the parent of $x_i$ is the parent of the parent of $x_{i-1}$. In what follows, we sometimes abuse notation and write $L(C)$ to denote the set $\{x_1,x_2,\ldots,x_l\}$.

Now, let $T$ and $T'$ be two phylogenetic $X$-trees. Each of the following reductions applied to $T$ and $T'$ results in two new phylogenetic trees $S$ and $S'$ with fewer leaves.

\begin{itemize}
    \item {\bf Subtree reduction:} Let $P$ be a maximal common pendant subtree of $T$ and $T'$ with at least two leaves. Obtain $S$ and $S'$  from $T$ and $T'$, respectively, by replacing $P$ with a single leaf with a new label that is not in $X$. Thus, a subtree reduction replaces a common pendant subtree with a single leaf.
    \item {\bf Chain reduction:} Let $C=(x_1,x_2,\ldots,x_l)$ be a maximal common chain of $T$ and $T'$ with $l\geq 4$. Obtain $S$ and $S'$ from $T|(X-\{x_4,x_5,\ldots,x_l\})$ and $T'|(X-\{x_4,x_5,\ldots,x_l\})$ respectively, by replacing leaf labels $x_1$, $x_2$, and $x_3$ with three new labels that are not in $X$. Thus, a chain reduction replaces a common chain of length at least four with a common chain of length three.
\end{itemize}

\noindent We next describe a $2$-step construction to obtain two phylogenetic trees from $T$ and $T'$ by repeated applications of the subtree and chain reductions. For $m\geq 0$, let $P_1,P_2,\ldots,P_m$ be distinct maximal pendant subtrees with at least two leaves that are common to $T$ and $T'$. Observe that $L(P_i)\cap L(P_j)=\emptyset$ for two distinct elements $i,j\in\{1,2,\ldots,m\}$. If $m=0$, then set $S=T$ and $S'=T'$ and, otherwise, obtain two phylogenetic $X'$-trees $S$ and $S'$  from $T$ and $T'$, respectively, by applying the subtree reduction to each $P_i$ such that each $P_i$ is replaced with a single leaf labeled $p_i$, where $p_i\notin X$. Set
$$X'=\big(X-\bigcup_{i=1}^m L(P_i)\big)\cup\{p_1,p_2,\ldots,p_m\}.$$  
Next, for $m'\geq 0$, let $C_1,C_2,\ldots,C_{m'}$ be distinct maximal chains of length at least four that are common to $S$ and $S'$. Again, by the maximality of each such chain, there exists no element in $X'$ that is a leaf of $C_i$ and $C_j$ for two distinct elements $i,j\in\{1,2,\ldots,m'\}$. If $m'=0$, set $R=S$ and $R'=S'$ and, otherwise, obtain two phylogenetic $X''$-trees $R$ and $R'$  from $S$ and $S'$, respectively, by applying the chain reduction to each $C_i=(x_1',x_2',\ldots,x_l')$ such that $x_1', x'_2, x'_3$ is replaced with $c_i^1, c^2_i, c^3_i$, respectively. Set
$$X''=\big(X'-\bigcup_{i=1}^{m'} L(C_i)\big) \cup \bigcup_{i=1}^{m'}\{c_i^1,c_i^2,c_i^3\}.$$ 
If at least one subtree or one chain reduction has been applied in the process of obtaining $R$ and $R'$ from $T$ and $T'$, respectively, then we refer to $R$ and $R'$ as a {\em reduced tree pair} with respect to $T$ and $T'$.
Moreover, if $R$ and $R'$ cannot be further reduced under the subtree or chain reduction, we refer to $R$ and $R'$ as a {\em fully reduced tree pair} with respect to $T$ and $T'$. Observe that a fully reduced tree pair with respect to $T$ and $T'$ can be obtained by applying the above $2$-step process so that all maximal pendant subtrees that are common to $T$ and $T'$ are reduced under the subtree reduction and, then, all maximal chains that are common to $S$ and $S'$ are reduced under the chain reduction.

Now let $\sigma''$ be an ordering on $X''$ such that, for each $i\in\{1,2,\ldots,m'\}$, the elements $c_i^1$, $c_i^2$, and $c_i^3$ have consecutive ranks under $\sigma''$ such that $\sigma''(c_i^1) < \sigma''(c_i^2) < \sigma''(c_i^3)$. 
Starting with $\sigma''$, obtain an ordering $\sigma'$ on $X'$ such that, for each $C_i=(x'_1,x'_2,\ldots,x'_l)$ with $i\in\{1,2,\ldots,m'\}$, the elements in $L(C_i)$ have consecutive ranks under $\sigma'$ with $\sigma'(x'_1)<\sigma'(x'_2)<\cdots < \sigma'(x_l)$ and at most one of the following holds for any two distinct elements $y$ and $y'$ in $X''$ with $\sigma''(y)<\sigma''(y')$:
\begin{enumerate}
    \itemsep 0pt
    \item If $y$ and $y'$ are both elements in $X'$, then $\sigma'(y)<\sigma'(y')$.
    \item If $y\in \{c^1_i, c^2_i, c^3_i\}$ and  $y'\in \{c^1_j, c^2_j, c^3_j\}$ with $i\ne j$, then $\sigma'(z)<\sigma'(z')$ for each pair $z$ and $z'$ of elements with $z\in L(C_i)$ and $z'\in L(C_j)$.
    \item If $y\in X'$ and $y'\in \{c^1_i, c^2_i, c^3_i\}$, then $\sigma'(y)<\sigma'(z)$ for each $z\in L(C_i)$.
    \item If $y\in \{c^1_i, c^2_i, c^3_i\}$ and $y'\in X'$, then $\sigma'(z)<\sigma'(y')$ for each $z\in L(C_i)$.
\end{enumerate} 
\noindent Lastly, obtain an ordering $\sigma$ on $X$ such that, for each $P_i$ with $i\in\{1,2,\ldots,m\}$, the elements in $L(P_i)$ have consecutive ranks under $\sigma$ and at most one of the following holds for any two distinct elements $y$ and $y'$ in $X'$ with $\sigma'(y)<\sigma'(y')$:
\begin{enumerate}
    \itemsep 0pt
    \item If $y$ and $y'$ are both elements in $X$, then $\sigma(y)<\sigma(y')$.
    \item If $y=p_i$ and  $y'=p_j$ with $i\ne j$, then $\sigma(z)<\sigma(z')$ for each pair $z$ and $z'$ of elements with $z\in L(P_i)$ and $z'\in L(P_j)$.
    \item If $y\in X$ and $y'=p_i$, then $\sigma'(y)<\sigma'(z)$ for each $z\in L(P_i)$.
    \item If $y=p_i$ and $y'\in X$, then $\sigma'(z)<\sigma'(y')$ for each $z\in L(P_i)$.
\end{enumerate} 
We refer to $\sigma$ (resp. $\sigma'$) as a {\em reduction preserving ordering} on $X$ (resp. $X'$) relative to $R$ and $R'$. 

\begin{lemma}\label{l:OLA-reduction}
    Let $T$ and $T'$ be two phylogenetic $X$-trees. Let $U$ and $U'$ be two  phylogenetic $X'$-trees that can be obtained from $T$ and $T'$, respectively, by a single application of the subtree or chain reduction. Let $\sigma'$ be an ordering on $X'$ such that, if $T$ and $T'$ have a maximal common chain $C$ that has been reduced to the chain  $(c^1,c^2,c^3)$ in $U$ and $U'$, then $c^1$, $c^2$, and $c^3$ have consecutive ranks with $\sigma'(c^1) < \sigma'(c^2) < \sigma'(c^3)$. Then, 
       \[d^{\sigma}_\OLA(T,T') = d^{\sigma'}_\OLA(S,S'), \]
    where $\sigma$ is a reduction preserving ordering on $X$ relative to $S$ and $S'$.
 \end{lemma}

\begin{proof}
    Let $T, T', U, U'$, $\sigma$, and $\sigma'$ be as stated in the lemma. Let $\mathbf{v}^{\sigma'}_U = [u_1, u_2, \ldots, u_{|X'|}]$ and $\mathbf{v}^{\sigma'}_{U'} = [u_1', u_2', \ldots, u_{|X'|}']$ be the OLA vectors of $U$ and $U'$ under $\sigma'$. By assumption, the Hamming distance between $\mathbf{v}^{\sigma'}_U$ and $\mathbf{v}^{\sigma'}_{U'}$ equals $d^{\sigma'}_\OLA(U,U')$. 

    Now, first consider the case that $U$ and $U'$ have been obtained from $T$ and $T'$ by a single application of the subtree reduction. In this case, let $P$ denote the common maximal pendant subtree of $T$ and $T'$, and let $p \in X'-X$ denote the unique leaf present in $U$ and $U'$ but not in $T$ and $T'$. Further, suppose that $\sigma'(p)=i$ for some $i \in \{1, \ldots, |X'|\}$. In particular, suppose that $u_i$ and $u_i'$ are the elements associated with $p$ in $\mathbf{v}^{\sigma'}_U$ and $\mathbf{v}^{\sigma'}_{U'}$, respectively. Since $\sigma$ is a reduction preserving ordering on $X$ relative to $U$ and $U'$, we now claim that we can obtain the OLA vectors for $T$ and $T'$ from the OLA vectors of $U$ and $U'$ by setting
    \begin{align*}
        \mathbf{v}^{\sigma}_T &= [u_1, u_2, \ldots, u_{i-1}, u_i, v_2, \ldots, v_{|L(P)|}, \widetilde{u}_{i+1}, \ldots, \widetilde{u}_{|X'|}] \quad \text{and } \\
        \mathbf{v}^{\sigma}_{T'} &= [u_1', u_2', \ldots, u_{i-1}', u_i', v_2, \ldots, v_{|L(P)|}, \widetilde{u}_{i+1}', \ldots, \widetilde{u}_{|X'|}'],
    \end{align*}
    where, the Hamming distance between $[\widetilde{u}_{i+1},\widetilde{u}_{i+2}, \ldots, \widetilde{u}_{|X'|}]$ and $[\widetilde{u}_{i+1}',\widetilde{u}_{i+2}', \ldots, \widetilde{u}_{|X'|}']$ equals that of $[u_{i+1},u_{i+2}, \ldots, u_{|X'|}]$ and  $[u_{i+1}',u_{i+2}', \ldots, u_{|X'|}']$.
    This is due to the following facts:
        \begin{enumerate}[label={\upshape(\roman*)}]
            \item All elements of $X$ whose ranks under $\sigma'$ are less than $\sigma'(p)$ are considered first and in the same order under $\sigma'$ and $\sigma$ when iteratively building $U$ and $U'$, respectively $T$ and $T'$, to obtain the OLA vectors. Thus, the first $i-1$ coordinates of $\mathbf{v}^{\sigma}_U$ and $\mathbf{v}^{\sigma}_T$ (resp. $\mathbf{v}^{\sigma}_{U'}$ and $\mathbf{v}^{\sigma}_{T'}$) coincide. 
            
            \item Now, consider the elements of $X \cap L(P)$ and let $x_p$ the leaf with minimal rank under $\sigma$ in this set, i.e., $\sigma(x_p) = \min_{x \in L(P)} \sigma(x)$. 
            Since $P$ is a common pendant subtree of $T$ and $T'$ and the elements of $L(P)$ have consecutive ranks under $\sigma$, it follows from Algorithm~\ref{alg:ola} that the vector coordinates for $x' \in X \cap L(P)- \{x_p\}$ are identical in $T$ and $T'$ and they correspond to positions $i+2, i+3,\ldots, i+|L(P)|-1$ of $\mathbf{v}^{\sigma}_T$ and $\mathbf{v}^{\sigma}_{T'}$, respectively. Now, the element associated with $x_p$ may differ between $\mathbf{v}^{\sigma}_T$ and $\mathbf{v}^{\sigma}_{T'}$; however, it clearly coincides with the element associated with $p$, namely $u_i$ (resp. $u_i'$) in $\mathbf{v}^{\sigma'}_U$ (resp. $\mathbf{v}^{\sigma'}_{U'}$).
            
            \item Finally, all elements of $X$ whose ranks under $\sigma'$ are greater than $\sigma'(p)$ are considered last and in the same order under $\sigma'$ and $\sigma$ when iteratively building $U$ and $U'$ (resp. $T$ and $T'$) to obtain the OLA vectors. Note that for each vertex $v$ of $T$ (resp. $T'$) that is introduced after the elements in $L(P)$ are added, we have $f^T_\OLA(v) = f^{U}_\OLA(u) + |L(P)|-1$ and $f^{T'}_\OLA(v) = f^{U'}_\OLA(u) + |L(P)|-1$, where $u$ is the vertex of $U$ (resp. $U'$) corresponding to $v$ in $T$ (resp. $T'$). Since the elements of $X$ whose ranks under $\sigma'$ are greater than $\sigma'(p)$ are clearly not added as siblings to vertices of $P$, we can conclude that $u_j = u_j'$ if and only if $\widetilde{u}_j' = \widetilde{u}_j'$ for each $j \in \{i+1,i+2,\ldots, |X'|\}$.  This implies that the Hamming distance between $[\widetilde{u}_{i+1},[\widetilde{u}_{i+2}, \ldots, \widetilde{u}_{|X'|}]$ and $[\widetilde{u}_{i+1}',\widetilde{u}_{i+2}', \ldots, \widetilde{u}_{|X'|}']$ equals that of $[u_{i+1},u_{i+2}, \ldots, u_{|X'|}]$ and  $[u_{i+1}', u_{i+2}', \ldots, u_{|X'|}']$.

        \end{enumerate} 
    In summary, the Hamming distance between $\mathbf{v}^{\sigma}_T$ and $\mathbf{v}^{\sigma}_{T'}$ equals the Hamming distance of $\mathbf{v}^{\sigma'}_U$ and $\mathbf{v}^{\sigma'}_{U'}$. Thus, when $U$ and $U'$ have been obtained from $T$ and $T'$ by a single application of the subtree reduction, $d^{\sigma}_\OLA(T,T') = d^{\sigma'}_\OLA(U,U')$ as claimed.

    Next, consider the case that $U$ and $U'$ have been obtained from $T$ and $T'$ by a single application of the chain reduction. In this case, let $C=(x_1, x_2, x_3, \ldots, x_{|L(C)|})$ denote the common chain of $T$ and $T'$, and let $c^1, c^2, c^3 \in X' - X$ denote the leaves present in $U$ and $U'$ but not in $T$ and $T'$. Let $i \in \{1, \ldots, |X'|-2\}$ be such that $u_i, u_{i+1}, u_{i+2}$ (resp. $u_i', u_{i+1}', u_{i+2}'$) are the elements associated with $c^1, c^2, c^3$ in $\mathbf{v}^{\sigma'}_U$ (resp. $\mathbf{v}^{\sigma'}_{U'}$). Since $\sigma$ is a reduction preserving ordering on $X$ relative to $U$ and $U'$, we now claim that we can obtain the OLA vectors for $T$ and $T'$ from the OLA vectors of $U$ and $U'$ by setting
    \begin{align*}
        \mathbf{v}^{\sigma}_T &= [u_1, u_2, \ldots, u_{i-1}, u_i, u_{i+1}, u_{i+2}, v_4, \ldots, v_{|L(C)|}, \widetilde{u}_{i+3}, \ldots, \widetilde{u}_{|X'|}] \quad \text{and } \\
        \mathbf{v}^{\sigma}_{T'} &= [u_1', u_2', \ldots, u_{i-1}', u_i', u_{i+1}', u_{i+2}', v_4, \ldots, v_{|L(C)|}, \widetilde{u}_{i+3}', \ldots, \widetilde{u}_{|X'|}'],
    \end{align*}
     where, the Hamming distance between $[\widetilde{u}_{i+3}, \widetilde{u}_{i+4}, \ldots, \widetilde{u}_{|X'|}]$ and $[\widetilde{u}_{i+3}',\widetilde{u}_{i+4}', \ldots, \widetilde{u}_{|X'|}']$ equals that of $[u_{i+3}, u_{i+4}, \ldots, u_{|X'|}]$ and  $[u_{i+3}',u_{i+4}', \ldots, u_{|X'|}']$.
     The reasoning is similar as in the previous case. In particular, the arguments presented in (i) and (iii) above for the elements of $X - L(C)$ immediately carry over. Thus, it remains to consider the elements of $X \cap L(C)$. Since $C = (x_1, x_2, x_3, \ldots, x_{|L(C)|})$ is a common chain of $T$ and $T'$ and, by construction, the elements of $L(C)$ have consecutive ranks under $\sigma$ with $\sigma(x_1) < \sigma(x_2) < \cdots < \sigma(x_{|L(C)|})$, it follows from Algorithm~\ref{alg:ola} that the elements of $\mathbf{v}^{\sigma}_T$ and $\mathbf{v}^{\sigma}_{T'}$ associated with $x_4,x_5, \ldots, x_{|L(C)|}$ coincide (specifically, the element associated with $x_j$ is $- \sigma(x_{j-1})$ for $j \in \{4,5, \ldots, |L(C)|\}$) in both vectors. Furthermore, the elements associated with $x_1, x_2, x_3$ may differ between $\mathbf{v}^{\sigma}_T$ and $\mathbf{v}^{\sigma}_{T'}$; however, they coincide with the elements associated with $c^1, c^2, c^3$ in $\mathbf{v}^{\sigma'}_U$ (resp. $\mathbf{v}^{\sigma'}_{U'}$), namely $u_i, u_{i+1}, u_{i+2}$ (resp. $u_i', u_{i+1}', u_{i+2}'$).
     
     This shows that the Hamming distance between $\mathbf{v}^{\sigma}_T$ and $\mathbf{v}^{\sigma}_{T'}$ equals the Hamming distance of $\mathbf{v}^{\sigma'}_U$ and $\mathbf{v}^{\sigma'}_{U'}$, thereby completing the proof.
\end{proof} 

We now recall Lemma 3.3 from~\cite{bordewichSemple2005}:
\begin{lemma}\label{l:reduction-size}
    Let $T$ and $T'$ be two phylogenetic $X$-trees. Let $R$ and $R'$ be a fully reduced tree pair with respect to $T$ and $T'$, and let $X'$ be the label set of $R$. Then 
    \[ \vert X' \vert \leq 28 \cdot d_\rSPR(T,T').\]
\end{lemma}

\begin{proof}[Proof of Theorem~\ref{thm:OLA-20SPR}]
Let $T$ and $T'$ be two phylogenetic $X$-trees, let $S$ and $S'$ be two phylogenetic $X'$-trees obtained from $T$ and $T'$ by reducing all maximal common pendant subtrees, and let $R$ and $R'$ be two phylogenetic $X''$-trees obtained from $S$ and $S'$ by reducing all maximal common chains. 
By Lemma~\ref{l:reduction-size}, $\vert X'' \vert \leq 28 \cdot d_\rSPR(T,T')$. 

Let $\sigma''$ be an ordering on $X''$ such that, if $S$ and $S'$ have a maximal common chain $C$ that has been reduced to the chain  $(c^1,c^2,c^3)$ in $R$ and $R'$, then $c^1$, $c^2$, and $c^3$ have consecutive ranks with $\sigma''(c^1) < \sigma''(c^2) < \sigma''(c^3)$. Further, let $\sigma'$ and $\sigma$ be orderings on $X'$ and $X$, respectively, satisfying the conditions stated above, such that $\sigma$ (resp. $\sigma'$) is a reduction preserving ordering on $X$ (resp. $X'$) relative to $R$ and $R'$.

We now repeatedly apply Lemma~\ref{l:OLA-reduction} with respect to chain reductions to conclude that 
\[ d^{\sigma'}_\OLA (S,S') = d^{\sigma''}_\OLA (R,R').\]
Next, we repeatedly apply Lemma~\ref{l:OLA-reduction} with respect to subtree reductions to conclude that 
\[ d^{\sigma}_\OLA (T,T') = d^{\sigma'}_\OLA (S,S').\]

In summary, this implies that based on an ordering $\sigma''$ on $X''$, for a reduction preserving ordering $\sigma$ on $X$ relative to $R$ and $R'$, we have
\[ d^{\sigma}_\OLA (T,T') =  d^{\sigma''}_\OLA (R,R').\]

Now, clearly $d^{\sigma''}_\OLA (R,R') \leq |X''|$, implying that
\[ d^{\sigma}_\OLA (T,T') =  d^{\sigma''}_\OLA (R,R') \leq |X''| \leq 28 \cdot d_\rSPR(T,T'),\]
which completes the proof.
\end{proof}

\begin{figure}[t]
\begin{center}
\begin{tikzpicture}[scale=.8]

    \node[fill=black,circle,inner sep=1pt, label=below:$x_1$]   at (0.5,0.5) {};
    \node[fill=black,circle,inner sep=1pt, label=below:$x_2$]   at (1.5,0.5) {};
    \node[fill=black,circle,inner sep=1pt, label=below:$x_3$]   at (2,1) {};
    \node[fill=black,circle,inner sep=1pt, label=below:$x_7$]   at (2.5,1.5) {};
    \node[fill=black,circle,inner sep=1pt, label=below:$x_8$]   at (3.5,1.5) {};
    \node[fill=black,circle,inner sep=1pt, label=below:$x_6$]   at (4,1) {};
    \node[fill=black,circle,inner sep=1pt, label=below:$x_4$]   at (4.5,0.5) {};
    \node[fill=black,circle,inner sep=1pt, label=below:$x_5$]   at (5.5,0.5) {};
    \node[fill=black,circle,inner sep=1pt,label=below:$x_9$]   at (4.5,3.5) {};
    \node[fill=black,circle,inner sep=1pt]   at (1,1) {};
    \node[fill=black,circle,inner sep=1pt]   at (1.5,1.5) {};
    \node[fill=black,circle,inner sep=1pt]   at (2,2) {};
    \node[fill=black,circle,inner sep=1pt]   at (3,3) {};
    \node[fill=black,circle,inner sep=1pt]   at (4,2) {};
    \node[fill=black,circle,inner sep=1pt]   at (5,1) {};
    \node[fill=black,circle,inner sep=1pt]   at (4.5, 1.5) {};
    \node[fill=black,circle,inner sep=1pt]   at (4,4) {};
    \node[fill=black,circle,inner sep=1pt, label=above:$\rho$]   at (4, 5) {};
    \draw (0.5,0.5) -- (4,4);
    \draw (5.5, 0.5) -- (3,3);
    \draw (1,1) -- (1.5, 0.5);
    \draw (1.5,1.5) -- (2,1);
    \draw (2,2) -- (2.5, 1.5);
    \draw (4,2) -- (3.5, 1.5);
    \draw (4.5,1.5) -- (4,1);
    \draw (5,1) -- (4.5,0.5);
    \draw (4,4) -- (4.5,3.5);
    \draw(4,5) -- (4,4);
    \node at (3,5.5) {\large $S$};

    \node[fill=black,circle,inner sep=1pt, label=below:$x_7$]   at (7.5,0.5) {};
    \node[fill=black,circle,inner sep=1pt, label=below:$x_8$]   at (8.5,0.5) {};
    \node[fill=black,circle,inner sep=1pt, label=below:$x_9$]   at (9,1) {};
    \node[fill=black,circle,inner sep=1pt, label=below:$x_1$]   at (9.5,1.5) {};
    \node[fill=black,circle,inner sep=1pt, label=below:$x_2$]   at (10,2) {};
    \node[fill=black,circle,inner sep=1pt, label=below:$x_3$]   at (10.5,2.5) {};
    \node[fill=black,circle,inner sep=1pt, label=below:$x_4$]   at (11,3) {};
    \node[fill=black,circle,inner sep=1pt, label=below:$x_5$]   at (11.5,3.5) {};
    \node[fill=black,circle,inner sep=1pt, label=below:$x_6$]   at (12,4) {};
    \node[fill=black,circle,inner sep=1pt]   at (8,1) {};
    \node[fill=black,circle,inner sep=1pt]   at (8.5,1.5) {};
    \node[fill=black,circle,inner sep=1pt]   at (9,2) {};
    \node[fill=black,circle,inner sep=1pt]   at (9.5,2.5) {};
    \node[fill=black,circle,inner sep=1pt]   at (10,3) {};
    \node[fill=black,circle,inner sep=1pt]   at (10.5,3.5) {};
    \node[fill=black,circle,inner sep=1pt]   at (11,4) {};
    \node[fill=black,circle,inner sep=1pt]   at (11.5,4.5) {};
    \node[fill=black,circle,inner sep=1pt, label=above:$\rho$]   at (11.5,5) {};
    \draw(11.5,5)--(11.5,4.5);
    \draw(11.5,4.5)--(7.5,0.5);
    \draw(8,1)--(8.5,0.5);
    \draw(8.5,1.5)--(9,1);
    \draw(9,2)--(9.5,1.5);
    \draw(9.5,2.5)--(10,2);
    \draw(10,3)--(10.5,2.5);
    \draw(10.5,3.5)--(11,3);
    \draw(11,4)--(11.5,3.5);
    \draw(11.5,4.5)--(12,4);
    \node at (10,5.5) {\large $S'$};

    \node[fill=black,circle,inner sep=1pt, label=below:$x_1$]   at (12.5,0.5) {};
    \node[fill=black,circle,inner sep=1pt, label=below:$x_2$]   at (13.5,0.5) {};
    \node[fill=black,circle,inner sep=1pt, label=below:$x_3$]   at (14,1) {};
    \node[fill=black,circle,inner sep=1pt, label=below:$x_4$]   at (14.5,1.5) {};
    \node[fill=black,circle,inner sep=1pt, label=below:$x_5$]   at (15,2) {};
    \node[fill=black,circle,inner sep=1pt, label=below:$x_6$]   at (15.5,2.5) {};
    \node[fill=black,circle,inner sep=1pt, label=below:$x_8$]   at (16,3) {};
    \node[fill=black,circle,inner sep=1pt, label=below:$x_7$]   at (16.5,3.5) {};
    \node[fill=black,circle,inner sep=1pt, label=below:$x_9$]   at (17,4) {};
    \node[fill=black,circle,inner sep=1pt]   at (13,1) {};
    \node[fill=black,circle,inner sep=1pt]   at (13.5,1.5) {};
    \node[fill=black,circle,inner sep=1pt]   at (14,2) {};
    \node[fill=black,circle,inner sep=1pt]   at (14.5,2.5) {};
    \node[fill=black,circle,inner sep=1pt]   at (15,3) {};
    \node[fill=black,circle,inner sep=1pt]   at (15.5,3.5) {};
    \node[fill=black,circle,inner sep=1pt]   at (16,4) {};
    \node[fill=black,circle,inner sep=1pt]   at (16.5,4.5) {};
    \node[fill=black,circle,inner sep=1pt, label=above:$\rho$]   at (16.5,5) {};
    \draw(16.5,5)--(16.5,4.5);
    \draw(16.5,4.5)--(12.5,0.5);
    \draw(13,1)--(13.5,0.5);
    \draw(13.5,1.5)--(14,1);
    \draw(14,2)--(14.5,1.5);
    \draw(14.5,2.5)--(15,2);
    \draw(15,3)--(15.5,2.5);
    \draw(15.5,3.5)--(16,3);
    \draw(16,4)--(16.5,3.5);
    \draw(16.5,4.5)--(17,4);
    \node at (15,5.5) {\large $S''$};
    
\end{tikzpicture}
\end{center}
\caption{Two phylogenetic $X$-trees $S$ and $S'$ with $d_\OLA^\ast(S,S') \geq 4$ (verified via exhaustive enumeration of all orderings), whereas $d_\rSPR(S,S') = 3$. Moreover, $d_\OLA^\ast(S,S'')=1$ and $d_\OLA^\ast(S',S'')=2$, implying that the OLA measure is not a distance since the triangle quality is violated for $S$, $S'$, and $S''$.}
\label{fig:OLA-rSPR}
\end{figure}
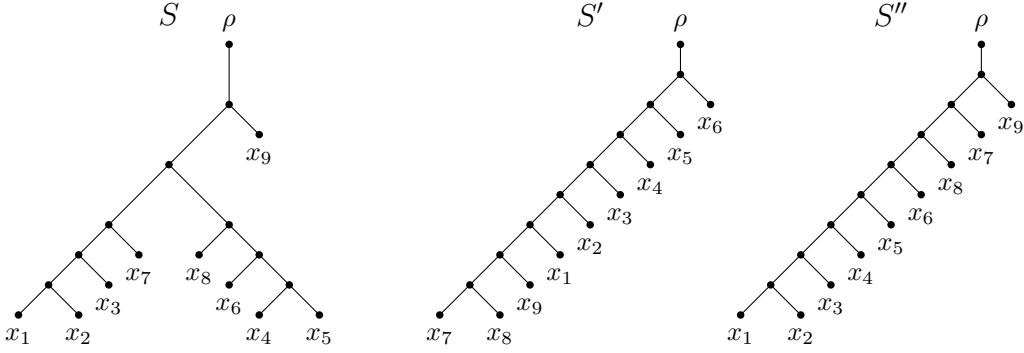

We end this section with three remarks on the relationship between the rSPR distance and the OLA and P2V measures, and an open problem raised in~\cite{ola}. First, although Theorem~\ref{thm:OLA-20SPR} shows that $d_\OLA^*(T,T')$ is bounded above by a function that is linear in $d_\rSPR(T,T')$ for any two phylogenetic trees $T$ and $T'$, there also exist pairs $S$ and $S'$ of phylogenetic trees such that $d_\rSPR(S,S')<d_\OLA^*(S,S')$. For example,  Figure~\ref{fig:OLA-rSPR} shows such a tree pair for which $3=d_\rSPR(S,S')<d_\OLA^*(S,S')=4$.

Second, for two phylogenetic $X$-trees $T$ and $T'$, Richman et al.~\cite[page 10]{ola} ask whether or not it is possible to bound the difference
$$\left|d_\OLA^\sigma(T,T')-d_\OLA^{\sigma'}(T,T')\right|,$$
where $\sigma$ and $\sigma'$ are two orderings  on $X$. We answer this question negatively by providing, for all $n\ge 1$, two phylogenetic trees on $n+1$ leaves for which the difference in OLA distances with respect to two distinct orderings can be as large as $n-2$. Consider the two phylogenetic $X$-trees $T$ and $T'$ with $|X|=n+1$ as shown in Figure~\ref{fig:reshuffling}. 
Recall, that we have $d_\OLA^\sigma(T,T')\leq n-1$ for any ordering $\sigma$ on $X$. Let $\sigma_1$ be the ordering on $X$ with 
$$\sigma_1(x_n)<\sigma_1(x_{n-1})<\cdots<\sigma_1(x_1)<\sigma_1(y),$$
and let $\sigma_{n-1}$ be the ordering on $X$ with $$\sigma_{n-1}(x_n)<\sigma_{n-1}(x_{n-1})<\sigma_{n-1}(y)<\sigma_{n-1}(x_{n-2})<\sigma_{n-1}(x_{n-3})<\cdots<\sigma_{n-1}(x_1).$$ It is straightforward to check that $d_\OLA^{\sigma_1}(T,T')=1$ and $d_\OLA^{\sigma_{n-1}}(T,T')=n-1$.
In general, for each $i\in\{1,2,\ldots,n-1\}$, there exists an ordering $\sigma_i$ on $X$ such that $d_\OLA^{\sigma_i}(T,T')=i$. More precisely, $\sigma_i$ is the unique ordering such that
\begin{enumerate}
\itemsep 0pt
\item $\sigma_i(y)=n+1-(i-1)$ and
\item $\sigma_i(x_n)<\sigma_i(x_{n-1})<\cdots<\sigma_i(x_1)$.
\end{enumerate}
Hence, for two arbitrary orderings $\sigma$ and $\sigma'$ on $X$, we have
$$\left|d^\sigma_\OLA(T, T') - d^{\sigma'}_\OLA(T, T')\right|\le n-2,$$
and this bound is sharp.

\begin{figure}
    \centering
    \begin{tikzpicture}[scale=.8]
    \node[fill=black,circle,inner sep=1pt, label=below:$x_1$]   at (0.5,0.5) {};
    \node[fill=black,circle,inner sep=1pt, label=below:$x_2$]   at (1.5,0.5) {};
    \node[fill=black,circle,inner sep=1pt, label=below:$x_3$]   at (2,1) {};
    \node[fill=black,circle,inner sep=1pt, label=below:$x_4$]   at (2.5,1.5) {};
    \node[fill=black,circle,inner sep=1pt, label=below:$x_{n-1}$]   at (3.5,2.5) {};
    \node[fill=black,circle,inner sep=1pt, label=below:$x_{n}$]   at (4,3) {};
    \node[fill=black,circle,inner sep=1pt, label=below:$y$]   at (4.5,3.5) {};
    \node[fill=black,circle,inner sep=1pt]   at (1,1) {};
    \node[fill=black,circle,inner sep=1pt]   at (1.5,1.5) {};
    \node[fill=black,circle,inner sep=1pt]   at (2,2) {};
    \node[fill=black,circle,inner sep=1pt]   at (4,4) {};
    \node[fill=black,circle,inner sep=1pt]   at (3.5,3.5) {};
    \node[fill=black,circle,inner sep=1pt]   at (3,3) {};
    \draw(0.5,0.5)--(2.25,2.25);
    \draw(2.25,2.25)--(2.75,2.75)[dotted];
    \draw(2.75,2.75)--(4,4);
    \draw(4,4)--(4,4.5);
    \draw(4,4)--(4.5,3.5);
    \draw(3.5,3.5)--(4,3);
    \draw(3,3)--(3.5,2.5);
    \draw(2,2)--(2.5,1.5);
    \draw(1.5,1.5)--(2,1);
    \draw(1,1)--(1.5,0.5); 
    \node[fill=black,circle,inner sep=1pt, label=above:$\rho$]   at (4,4.5) {};
    \node at (3,4.5) {\large $T$};

    \node[fill=black,circle,inner sep=1pt, label=below:$y$]   at (4.5,0.5) {};
    \node[fill=black,circle,inner sep=1pt, label=below:$x_1$]   at (5.5,0.5) {};
    \node[fill=black,circle,inner sep=1pt, label=below:$x_2$]   at (6,1) {};
    \node[fill=black,circle,inner sep=1pt, label=below:$x_3$]   at (6.5,1.5) {};
    \node[fill=black,circle,inner sep=1pt, label=below:$x_{n-2}$]   at (7.5,2.5) {};
    \node[fill=black,circle,inner sep=1pt, label=below:$x_{n-1}$]   at (8,3) {};
    \node[fill=black,circle,inner sep=1pt, label=below:$x_n$]   at (8.5,3.5) {};
    \node[fill=black,circle,inner sep=1pt]   at (5,1) {};
    \node[fill=black,circle,inner sep=1pt]   at (5.5,1.5) {};
    \node[fill=black,circle,inner sep=1pt]   at (6,2) {};
    \node[fill=black,circle,inner sep=1pt]   at (8,4) {};
    \node[fill=black,circle,inner sep=1pt]   at (7.5,3.5) {};
    \node[fill=black,circle,inner sep=1pt]   at (7,3) {};
    \draw(4.5,0.5)--(6.25,2.25);
    \draw(6.25,2.25)--(6.75,2.75)[dotted];
    \draw(6.75,2.75)--(8,4);
    \draw(8,4)--(8,4.5);
    \draw(8,4)--(8.5,3.5);
    \draw(7.5,3.5)--(8,3);
    \draw(7,3)--(7.5,2.5);
    \draw(6,2)--(6.5,1.5);
    \draw(5.5,1.5)--(6,1);
    \draw(5,1)--(5.5,0.5); 
    \node[fill=black,circle,inner sep=1pt, label=above:$\rho$]   at (8,4.5) {};
    \node at (7,4.5) {\large $T'$};
    \end{tikzpicture}

    \caption{Two phylogenetic $X$-trees $T$ and $T'$ with $|X|=n+1$ for which the OLA distance differs by as much as $n-2$, depending on the ordering.}
    \label{fig:reshuffling}
\end{figure}
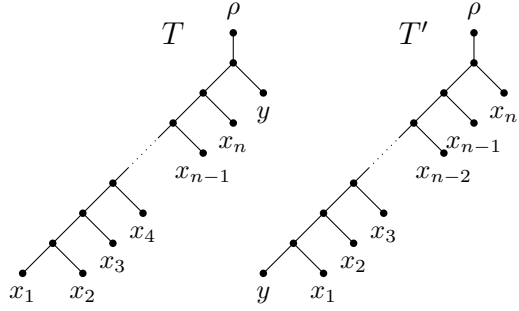

Third, we briefly turn to the $\PV$ measure of two phylogenetic trees $T$ and $T'$ and show that, similar to OLA, there is no clear relationship between $d_\rSPR(T,T')$ and $d_\PV^*(T,T')$ because $d_\PV^*(T,T')$ can be strictly greater or strictly smaller than $d_\rSPR(T,T')$. Figure~\ref{fig:P2V-SPR} shows an example of two phylogenetic trees $T$ and $T'$ with $2=d_\rSPR(T,T')>d_\PV^*(T,T')=1$, but also shows an example of two phylogenetic trees $S$ and $S'$ with $2=d_\rSPR(S,S')<d_\PV^*(S,S')=3$.

Despite the somewhat negative results presented in this section, we will see in the next section that there is a direct relationship between $h(T,T')$ and $d_\HOP^*(T,T')$ for any two phylogenetic trees $T$ and $T'$ on the same label set. Indeed, we will show that both measures are equivalent.

\begin{figure}[t]
\begin{center}
    \begin{tikzpicture}[scale=.8]

     \node[fill=black,circle,inner sep=1pt, label=below:$x_1$]   at (0.5,0.5) {};
     \node[fill=black,circle,inner sep=1pt, label=below:$x_3$]   at (1.5,0.5) {};
     \node[fill=black,circle,inner sep=1pt, label=below:$x_4$]   at (2,1) {};
     \node[fill=black,circle,inner sep=1pt, label=below:$x_2$]   at (2.5,1.5) {};
     \node[fill=black,circle,inner sep=1pt]   at (1,1) {};
     \node[fill=black,circle,inner sep=1pt]   at (1.5,1.5) {};
     \node[fill=black,circle,inner sep=1pt]   at (2,2) {};
     \node[fill=black,circle,inner sep=1pt, label=above:$\rho$]   at (2,2.5) {};
     \draw(2,2.5)--(2,2);
     \draw(2,2)--(0.5,0.5);
     \draw(2,2)--(2.5,1.5);
     \draw(1.5,1.5)--(2,1);
     \draw(1,1)--(1.5,0.5);
     \node at (1,3) {\large $T$};

     \node[fill=black,circle,inner sep=1pt, label=below:$x_2$]   at (3,0.5) {};
     \node[fill=black,circle,inner sep=1pt, label=below:$x_3$]   at (4,0.5) {};
     \node[fill=black,circle,inner sep=1pt, label=below:$x_4$]   at (4.5,1) {};
     \node[fill=black,circle,inner sep=1pt, label=below:$x_1$]   at (5,1.5) {};
     \node[fill=black,circle,inner sep=1pt]   at (3.5,1) {};
     \node[fill=black,circle,inner sep=1pt]   at (4,1.5) {};
     \node[fill=black,circle,inner sep=1pt]   at (4.5,2) {};
     \node[fill=black,circle,inner sep=1pt, label=above:$\rho$]   at (4.5,2.5) {};
     \draw(4.5,2.5)--(4.5,2);
     \draw(4.5,2)--(3,0.5);
     \draw(4.5,2)--(5,1.5);
     \draw(4,1.5)--(4.5,1);
     \draw(3.5,1)--(4,0.5);
     \node at (3.5,3) {\large $T'$};

     \node[fill=black,circle,inner sep=1pt, label=below:$x_1$]   at (6.5,-0.5) {};
     \node[fill=black,circle,inner sep=1pt, label=below:$x_2$]   at (7.5,-0.5) {};
     \node[fill=black,circle,inner sep=1pt, label=below:$x_3$]   at (8,0) {};
     \node[fill=black,circle,inner sep=1pt, label=below:$x_4$]   at (8.5,0.5) {};
     \node[fill=black,circle,inner sep=1pt, label=below:$x_5$]   at (9,1) {};
     \node[fill=black,circle,inner sep=1pt, label=below:$x_6$]   at (9.5,1.5) {};
     \node[fill=black,circle,inner sep=1pt]   at (7,0) {};
     \node[fill=black,circle,inner sep=1pt]   at (7.5,0.5) {};
     \node[fill=black,circle,inner sep=1pt]   at (8,1) {};
     \node[fill=black,circle,inner sep=1pt]   at (8.5,1.5) {};
     \node[fill=black,circle,inner sep=1pt]   at (9.5,1.5) {};
     \node[fill=black,circle,inner sep=1pt]   at (9,2) {};
     \node[fill=black,circle,inner sep=1pt, label=above:$\rho$]   at (9,2.5) {};
     \draw(9,2.5)--(9,2);
     \draw(9,2)--(6.5,-0.5);
     \draw(9,2)--(9.5,1.5);
     \draw(8.5,1.5)--(9,1);
     \draw(8,1)--(8.5,0.5);
     \draw(7,0)--(7.5,-0.5);
     \draw(7.5,0.5)--(8,0);
     \node at (8,3) {\large $S$};

     \node[fill=black,circle,inner sep=1pt, label=below:$x_4$]   at (9.5,-0.5) {};
     \node[fill=black,circle,inner sep=1pt, label=below:$x_5$]   at (10.5,-0.5) {};
     \node[fill=black,circle,inner sep=1pt, label=below:$x_6$]   at (11,0) {};
     \node[fill=black,circle,inner sep=1pt, label=below:$x_1$]   at (11.5,0.5) {};
     \node[fill=black,circle,inner sep=1pt, label=below:$x_2$]   at (12,1) {};
     \node[fill=black,circle,inner sep=1pt, label=below:$x_3$]   at (12.5,1.5) {};
     \node[fill=black,circle,inner sep=1pt]   at (10,0) {};
     \node[fill=black,circle,inner sep=1pt]   at (10.5,0.5) {};
     \node[fill=black,circle,inner sep=1pt]   at (11,1) {};
     \node[fill=black,circle,inner sep=1pt]   at (11.5,1.5) {};
     \node[fill=black,circle,inner sep=1pt]   at (12.5,1.5) {};
     \node[fill=black,circle,inner sep=1pt]   at (12,2) {};
     \node[fill=black,circle,inner sep=1pt, label=above:$\rho$]   at (12,2.5) {};
     \draw(12,2.5)--(12,2);
     \draw(12,2)--(9.5,-0.5);
     \draw(12,2)--(12.5,1.5);
     \draw(11.5,1.5)--(12,1);
     \draw(11,1)--(11.5,0.5);
     \draw(10,0)--(10.5,-0.5);
     \draw(10.5,0.5)--(11,0);
     \node at (11,3) {\large $S'$};

     \node[fill=black,circle,inner sep=1pt, label=below:$x_4$]   at (14.5,-0.5) {};
     \node[fill=black,circle,inner sep=1pt, label=below:$x_5$]   at (15.5,-0.5) {};
     \node[fill=black,circle,inner sep=1pt, label=below:$x_6$]   at (16,0) {};
     \node[fill=black,circle,inner sep=1pt, label=below:$x_3$]   at (18,0) {};
     \node[fill=black,circle,inner sep=1pt, label=below:$x_2$]   at (18.5,-0.5) {};
     \node[fill=black,circle,inner sep=1pt, label=below:$x_1$]   at (19.5,-0.5) {};
     \node[fill=black,circle,inner sep=1pt]   at (15,0) {};
     \node[fill=black,circle,inner sep=1pt]   at (19,0) {};
     \node[fill=black,circle,inner sep=1pt]   at (18.5,0.5) {};
     \node[fill=black,circle,inner sep=1pt]   at (18,0) {};
     \node[fill=black,circle,inner sep=1pt]   at (15.5,0.5) {};
     \node[fill=black,circle,inner sep=1pt]   at (17,2) {};
     \node[fill=black,circle,inner sep=1pt, label=above:$\rho$]   at (17,2.5) {};
     \draw(17,2.5)--(17,2);
     \draw(17,2)--(14.5,-0.5);
     \draw(17,2)--(19.5,-0.5);
     \draw(15.5,0.5)--(16,0);
     \draw(15,0)--(15.5,-0.5);
     \draw(19,0)--(18.5,-0.5);
     \draw(18.5,0.5)--(18,0);
     \node at (16,3) {\large $R$};
     
    \end{tikzpicture}
\end{center}
\caption{Left: Two phylogenetic $X$-trees $T$ and $T'$ with $d_\rSPR(T,T')=2$ and $d^\ast_\PV(T,T') = 1$. Middle: Two phylogenetic $X$-trees $S$ and $S'$ with $d_\rSPR(S,S')=2$ and $d^\ast_\PV(S,S') = 3$. Middle and right: Three phylogenetic $X$-trees $S$, $S'$, and $R$ with $h(S,S')= d^\ast_\HOP(S,S')=3, h(S,R)=d^\ast_\HOP(S,R)=1$, and $h(S',R)= d^\ast_\HOP(S', R) =1$, showing that the HOP measure is not a distance. Similarly, as $d^\ast_\PV(S,R)= d^\ast_\PV(S',R)=1$, whereas $d^\ast_\PV(S,S') = 3$, the P2V measure is also not a distance.}
\label{fig:P2V-SPR}
\end{figure}
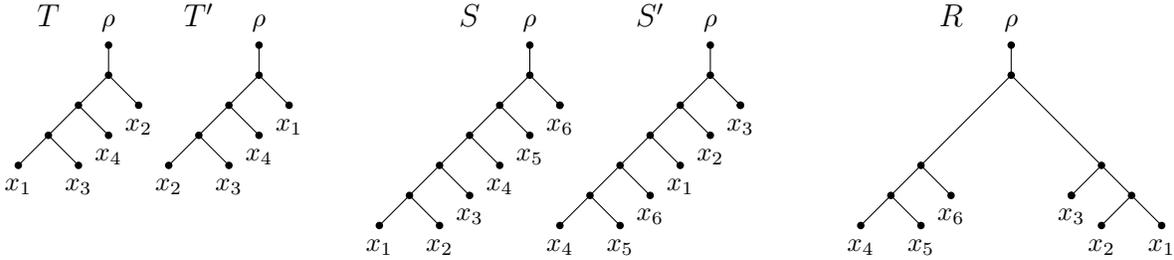

\subsection{Bounding by the hybrid number}
\label{results:hybrid}

In this section, we relate each of the three order-dependent measures to the hybrid number. Specifically, for two phylogenetic $X$-trees $T$ and $T'$, we establish that $d^*_\HOP(T,T')$ and $h(T,T')$ are equivalent, which also proves the conjecture of \cite[Remark~3, p.~9]{hop} concerning the computational complexity of computing $d^*_\HOP(T,T')$. Further, we show that $d^*_\OLA(T,T')$ and  $d^*_\PV(T,T')$ are bounded above by the hybrid number. However, unlike the HOP measure, there are phylogenetic trees for which the hybrid number is strictly larger than the  OLA and P2V measures.

\begin{theorem}\label{thm:hybrid-upper-bound}
Let $T$ and $T'$ be two  phylogenetic $X$-trees. Then, for each $\Theta\in\{\HOP,\OLA,\PV\}$,
$$d^*_\Theta(T, T')\le h(T, T').$$
\end{theorem}

We note that the hybrid number is an upper bound on the rSPR distance between two phylogenetic $X$-trees and that both measures can differ by up to $n-\lceil 2\sqrt n\,\rceil$~\cite{humphries2009note} where $n = |X|$.  As such, Theorem~\ref{thm:hybrid-upper-bound} does not further strengthen the results presented in Section~\ref{results:rspr}.

Now to prove Theorem~\ref{thm:hybrid-upper-bound}, we require a few additional lemmas. 
The next two lemmas show that, if two phylogenetic trees $T$ and $T'$ with label set $X$ have a common pendant subtree $S$, then the HOP and OLA labelings of $S$ in $T$ and $T'$ are identical for any fixed ordering on $X$.

\begin{lemma}
Let $T$ and $T'$ be two phylogenetic $X$-trees, where $|X|=n$, and let $\sigma$ be an ordering of the elements in $X$. Suppose that $T$ and $T'$ have a common pendant subtree $S$. Then, the $\HOP$ labeling of $S$ under $\sigma$ is identical for $T$ and $T'$.
\label{lem:subtree-labeling1}
\end{lemma}

\begin{proof}
The proof is by induction on $n$. If $n=1$, then the result follows immediately. 

Now assume that the result holds for all pairs of phylogenetic trees that have at most $n-1$ leaves. Let $x$ be the element in $X$ such that $\sigma(x)=n$. Let $T_1=T|(X- \{x\})$ and let $T_1'=T'|(X- \{x\})$. If $x\in L(S)$, set $S_1=S|(X- \{x\})$ and, otherwise, set $S_1=S$. Since $S$ is a common pendant subtree of $T$ and $T'$, it follows that  $S_1$ is a common pendant subtree of $T_1$ and $T_1'$. Moreover, by the induction assumption, the $\HOP$ labeling of $S_1$ under $\sigma_{-x}$ is identical for $T_1$ and $T_1'$. Let $e$ and $e'$ be the unique edge in $T_1$ and $T_1'$, respectively, such that $T$ can be obtained from $T_1$ by subdividing $e$ with a new vertex $u$ and adding the new edge $(u,x)$ and, similarly, $T'$ can be obtained from $T_1'$ by subdividing $e'$ with a new vertex $u'$ and adding the new edge $(u',x)$. Then, the HOP labeling of $T$ and $T'$ can be obtained from that of $T_1$ and $T_1'$, respectively, by setting $f_\HOP^T(x)=f_\HOP^{T'}(x)=\ul{\sigma(x)}$ and $f_\HOP^T(u)=f_\HOP^{T'}(u')=\sigma(x)$,
thereby implying that the HOP labeling of $S$ under $\sigma$ is identical in $T$ and $T'$. 
\end{proof}

\begin{lemma}
Let $T$ and $T'$ be two phylogenetic $X$-trees, where $|X|=n$, and let $\sigma$ be an ordering of the elements in $X$. Furthermore, let $\mathbf{v}$ and $\mathbf{v}'$  be the $\OLA$ vectors for  $T$ and $T'$, respectively, under $\sigma$.
Suppose that $T$ and $T'$ have a common pendant subtree $S$. Then, the $\OLA$  labeling of $S$ under $\sigma$ is identical for $T$ and $T'$. Moreover, $\mathbf{v}$ and $\mathbf{v}'$  restricted to  $L(S)$ have Hamming distance at most $1$.
\label{lem:subtree-labeling2}
\end{lemma}

\begin{proof}
We establish both parts using induction on $n$. If $n=1$, then the result follows since $T\simeq T'$. 

Now assume that $n>1$ and that the result holds for all pairs of phylogenetic $X$-trees that have at most $n-1$ leaves. Since the lemma holds whenever $|L(S)|=1$, we may assume for the remainder of the proof that $L(S)\geq 2$. Let $x$ be the element in $X$ such that $\sigma(x)=n$. Let $T_1=T|(X-\{x\})$, let $T_1'=T'|(X-\{x\})$, and let $S_1=S|(L(S)- \{x\})$. Observe that $S_1=S$ if $x\notin L(S)$. Since $S$ is a common pendant subtree of $T$ and $T'$, it follows that $S_1$ is a common pendant subtree of $T_1$ and $T_1'$.  Let $\mathbf{v}_1=[v_1,v_2,\ldots,v_{n-1}]$ and $\mathbf{v}_1'=[v_1',v_2',\ldots,v_{n-1}']$  be the $\OLA$ vectors for  $T_1$ and $T'_1$, respectively, under $\sigma_{-x}$.  Then the following two statements  follow from the induction assumption. First, the $\OLA$ labeling of $S_1$ is identical for $T_1$ and $T_1'$ under $\sigma_{-x}$. Second, $\mathbf{v}_1$ and $\mathbf{v}_1'$ restricted to the elements in $L(S_1)$ have Hamming distance at most 1. Now, let $e=(u,w)$ and $e'=(u',w')$ be the unique edge in $T_1$ and $T_1'$, respectively, such that $T$ can be obtained from $T_1$ by subdividing $e$ with a new vertex $v$ and adding the new edge $(v,x)$ and, similarly, $T'$ can be obtained from $T_1'$ by subdividing $e'$ with a new vertex $v'$ and adding the new edge $(v',x)$. We can view the vertex set of $T$ as the union of the vertex set of $T_1$ and $\{v,x\}$ and, similarly for $T'$.  The OLA labeling of $T$ and $T'$ under $\sigma$ can be obtained from that of $T_1$ and $T_1'$, respectively, under $\sigma_{-x}$ by setting $f^T_\OLA(x)=f^{T'}_\OLA(x)=n$, setting $f^T_\OLA(v)=f^{T'}_\OLA(v')=-\sigma(x)$ and, for each vertex $t$ (resp. $t'$) in $T_1$ (resp. $T_1'$), setting $f^T_\OLA(t) = f^{T_1}_\OLA(t)$ (resp. $f^T_\OLA(t') = f^{T_1}_\OLA(t')$). Since $S$ is a common pendant subtree of $T$ and $T'$ and the $\OLA$ labeling of $S_1$ under $\sigma_{-x}$ is identical for $T_1$ and $T_1'$, it  follows that the OLA labeling of $S$ under $\sigma$ is also identical in $T$ and $T'$. Now, by definition of $\mathbf{v}$  and $\mathbf{v}'$, recall that  $\mathbf{v}=[v_1,v_2,\ldots,v_{n-1},v_n]$ and $\mathbf{v}'=[v_1',v_2',\ldots,v_{n-1}', v_n']$. Furthermore, since the coordinate is determined by the sibling, $v_n=f^T_\OLA(w)$ and $v_n'=f^{T'}_\OLA(w')$. 
If $x\notin L(S)$, then it immediately follows that $\mathbf{v}$ and $\mathbf{v}'$ restricted to $L(S)$ also have Hamming distance at most 1. On the other hand, if $x\in L(S)$, then  $f^T_\OLA(w)=f^{T'}_\OLA(w')$ since the OLA labeling of $S$ under $\sigma$ is identical in $T$ and $T'$. Thus, as $\mathbf{v}_1$ and $\mathbf{v}_1'$ restricted to the elements in $L(S_1)$ have Hamming distance at most 1, it again follows that $\mathbf{v}$ and $\mathbf{v}'$ restricted to $L(S)$ also have Hamming distance at most 1. 
\end{proof}

The analogue of Lemma~\ref{lem:subtree-labeling2} for $\PV$ does not hold. 
However, the following weaker version holds. 

\begin{lemma}
Let $T$ and $T'$ be two phylogenetic $X$-trees, where $|X|=n$. Suppose that $T$ and $T'$ have a common pendant subtree $S$.  Let $\sigma$
be an ordering of the elements in $X$ such that the elements in $X- L(S)$ are bijectively mapped to the elements in $\{1,2,\ldots,n-|L(S)|\}$. 
Let $\mathbf{v}$ and $\mathbf{v}'$  be the $\PV$ vectors for  $T$ and $T'$, respectively, under $\sigma$.
Then the $\PV$  labeling of $S$ under $\sigma$ is identical for $T$ and $T'$. Moreover, $\mathbf{v}$ and $\mathbf{v}'$ restricted to  $L(S)$ have Hamming distance at most $1$.
\label{lem:subtree-labeling3}
\end{lemma}

\begin{proof}
We first show that the $\PV$  labeling of $S$ under $\sigma$ is identical for $T$ and $T'$. By the choice of $\sigma$, there exists an ordering, $t_1,t_2,\ldots,t_{|L(S)|-1}$, on the internal vertices of $T|{L(S)}$ such that $f_\PV^T(t_i)=n+i$ for each $i\in\{1,2,\ldots, |L(S)|-1\}$.  Since $S$ is a common pendant subtree of $T$ and $T'$ it is now straightforward to check that the $\PV$  labeling of $S$ under $\sigma$ is identical for $T$ and $T'$. 

To complete the proof, we show by induction on $|L(S)|$ that $\mathbf{v}$ and $\mathbf{v}'$ restricted to  $L(S)$ have Hamming distance at most 1.
If $|L(S)|=1$, then the result clearly follows. 

Now assume that $|L(S)|>1$ and that the result holds for all pairs of phylogenetic trees that have a common pendant subtree whose size is strictly less than $|L(S)|$. Let $x$ be the element in $X$ such that $\sigma(x)=n$.\ By the choice of $\sigma$, we have 
$x\in L(S)$.
Let $T_1=T|(X- \{x\})$, let $T_1'=T'|(X- \{x\})$, and let $S_1=S|(L(S)- \{x\})$. Since $S$ is a common pendant subtree of $T$ and $T'$ and $|L(S)-\{x\}|\ge 1$, it follows that $S_1$ is a common pendant subtree of $T_1$ and $T_1'$.  Let $\mathbf{v}_1=[v_1,v_2,\ldots,v_{n-1}]$ and $\mathbf{v}_1'=[v_1',v_2',\ldots,v_{n-1}']$  be the $\PV$ vectors for  $T_1$ and $T'_1$, respectively, under $\sigma_{-x}$.  Then, by the induction assumption, we have that $\mathbf{v}_1$ and $\mathbf{v}_1'$ restricted to the elements in $L(S_1)$ have Hamming distance at most 1. Now, let $e=(u,w)$ and $e'=(u',w')$ be the unique edge in $T_1$ and $T_1'$, respectively, such that $T$ can be obtained from $T_1$ by subdividing $e$ with a new vertex $v$ and adding the new edge $(v,x)$ and, similarly, $T'$ can be obtained from $T_1'$ by subdividing $e'$ with a new vertex $u'$ and adding the new edge $(u',x)$. Then $\mathbf{v}=[v_1,v_2,\ldots,v_n]$ and $\mathbf{v}'=[v_1',v_2',\ldots,v_n']$ with $v_n=f^{T}_\PV(w)$ and $v_n'=f^{T'}_\PV(w')$. Since $w$ and $w'$ is a vertex of $T|L(S)$ and $T'|{L(S)}$, respectively, and  the $\PV$ labeling of $S$ is identical for $T$ and $T'$ under $\sigma$, we have  $v_n=v_n'=f_\PV^{T}(w)$. Hence, as $\mathbf{v}_1$ and $\mathbf{v}_1'$ restricted to the elements in $L(S_1)$ have Hamming distance at most 1, it now follows that $\mathbf{v}$ and $\mathbf{v}'$ restricted to $L(S)$ also have Hamming distance at most 1. \end{proof}

We are now in the position to prove Theorem~\ref{thm:hybrid-upper-bound}.

\begin{proof}[Proof of Theorem~\ref{thm:hybrid-upper-bound}]
The proof is by induction on $n$. If $n \leq 2$, then $T$ is isomorphic to $T'$, and the statement immediately holds. 

Suppose that $n\ge 3$ and that the statement holds for all pairs of phylogenetic trees with at most $n-1$ leaves.
Let $F=\{T_{\rho}, T_1, T_2, \ldots, T_k\}$ be an acyclic agreement forest for $T$ and $T'$. Since $F$ is acyclic, there is a tree in $F$ that is pendant in $T$ and $T'$. Without loss of generality, we may assume that this tree is $T_k$. Let $n_k=|L(T_k)|$, let $T_1=T|(X- L(T_k))$, and let $T'_1=T'|(X- L(T_k))$. Furthermore, let $F_1=\{T_{\rho}, T_1, T_2, \ldots, T_{k-1}\}$. Since $F$ is an acyclic agreement forest for $T$ and $T'$, it follows that $F_1$ is an acyclic agreement forest for $T_1$ and $T'_1$. By the induction assumption and Theorem~\ref{thm:hybridnumber}, there is an ordering $\sigma_1$ on $X - L(T_k)$ such that
\begin{equation}\label{eq:one}
d_\Theta^*(T_1,T_1')\leq d^{\sigma_1}_\Theta(T_1, T'_1)\leq k-1.
\end{equation}
Let $\sigma$ be an ordering on $X$ such that $\sigma(x)=\sigma_1(x)$ for each $x\in X-L(T_k)$. Then $\sigma(y)>n-n_k$ for each $y\in L(T_k)$. Let $p$ and $p'$ denote the parent of the root of $T_k$ in $T$ and $T'$, respectively.  

We next complete the induction step for $\Theta=\HOP$. Since  $T_k$ is a pendant subtree of $T$ and $T'$ it follows  that $p$ and $p'$ have the same label which, by construction, is $n-n_k+1$.  Moreover, by Lemma~\ref{lem:subtree-labeling2}, the HOP labeling of the  vertices of $T$ and $T'$ that correspond to $T_k$  under $\sigma$ are identical. Hence each element in $\{n-n_k+2,n-n_k+3,\ldots,n\}$ is contained in a longest  common subsequence of the HOP sequences for $T$ and $T'$. It now follows that $$d^*_\HOP(T, T')\leq d^{\sigma}_{\rm HOP}(T, T')  \leq d^{\sigma_1}_{\rm HOP}(T_1, T'_1) + 1 \le (k-1) + 1 = k,$$ where the second inequality is strict if $n-n_k+1$ is an element of a longest common subsequence of the HOP sequences for $T$ and $T'$. 
Choosing $F$ to be a maximum acyclic agreement forest establishes the lemma for $\Theta=\HOP$.

Next let $\Theta\in\{\OLA,\PV\}$, and let $\mathbf{v}_1=[v_1,v_2,\ldots, v_{n-n_k}]$ and $\mathbf{v}_1'=[v_1',v'_2,\ldots, v'_{n-n_k}]$ be the $\Theta$ vectors for $T_1$ and $T_1'$, respectively, under $\sigma_1$. It follows from Equation~\eqref{eq:one} that the Hamming distance between $\mathbf{v}_1$ and $\mathbf{v}_1'$ is at most $k-1$. Moreover, by the choice of $\sigma$, the $\Theta$ vectors for $T$ and $T'$ are  $\mathbf{v}=[v_1,v_2,\ldots, v_{n-n_k},v_{n-n_k+1},\ldots,v_n]$ and $\mathbf{v}'=[v_1',v'_2,\ldots, v'_{n-n_k},v'_{n-n_k+1},\ldots,v'_n]$. Hence, by Lemma~\ref{lem:subtree-labeling2} for $\Theta=\OLA$ and Lemma~\ref{lem:subtree-labeling3} for $\Theta=\PV$, $\mathbf{v}$ and $\mathbf{v}'$ restricted to $L(T_k)$ have Hamming distance at most 1. In turn, this implies that 
$$d^*_\Theta(T, T')\leq d^{\sigma}_{\Theta}(T, T')  =d^{\sigma_1}_{\Theta}(T_1, T'_1) + 1 \le (k-1) + 1 = k$$
Again choosing $F$ to be a maximum acyclic agreement forest establishes the lemma for $\Theta\in\{\OLA,\PV\}$.
\end{proof}

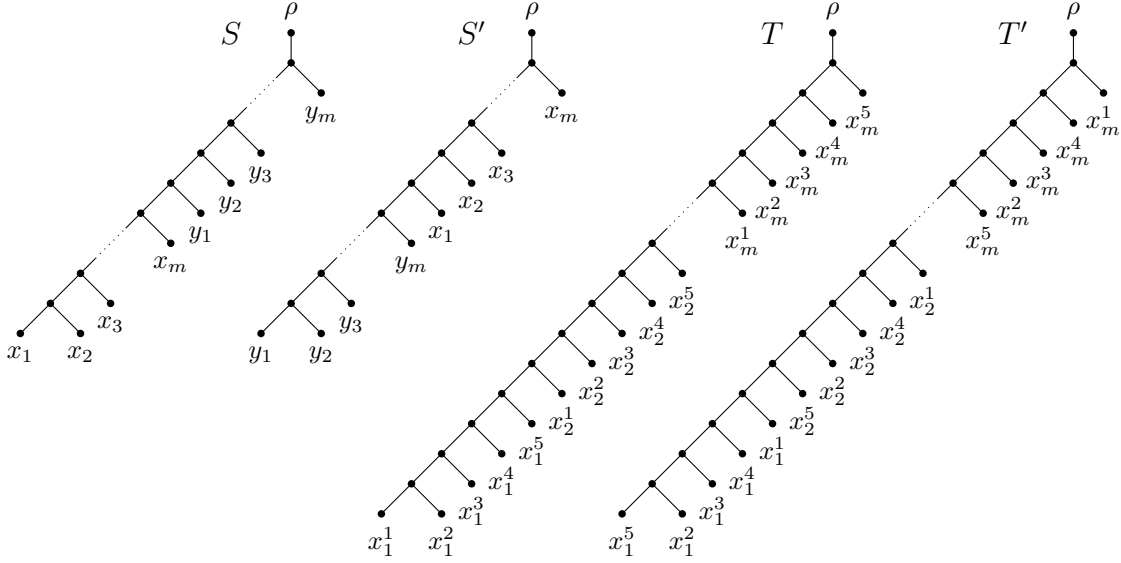
\begin{figure}[t]
    \centering
    \begin{tikzpicture}[scale=.8]

    \node[fill=black,circle,inner sep=1pt, label=below:$x_1$]   at (0.5,0.5) {};
    \node[fill=black,circle,inner sep=1pt, label=below:$x_2$]   at (1.5,0.5) {};
    \node[fill=black,circle,inner sep=1pt, label=below:$x_3$]   at (2,1) {};
    \node[fill=black,circle,inner sep=1pt, label=below:$x_m$]   at (3,2) {};
    \node[fill=black,circle,inner sep=1pt, label=below:$y_1$]   at (3.5,2.5) {};
    \node[fill=black,circle,inner sep=1pt, label=below:$y_2$]   at (4,3) {};
    \node[fill=black,circle,inner sep=1pt, label=below:$y_3$]   at (4.5,3.5) {};
    \node[fill=black,circle,inner sep=1pt, label=below:$y_m$]   at (5.5,4.5) {};
    \node[fill=black,circle,inner sep=1pt]   at (1,1) {};
    \node[fill=black,circle,inner sep=1pt]   at (1.5,1.5) {};
    \node[fill=black,circle,inner sep=1pt]   at (2.5,2.5) {};
    \node[fill=black,circle,inner sep=1pt]   at (3,3) {};
    \node[fill=black,circle,inner sep=1pt]   at (3.5,3.5) {};
    \node[fill=black,circle,inner sep=1pt]   at (4,4) {};
    \node[fill=black,circle,inner sep=1pt]   at (5,5) {};
    \draw(0.5,0.5)--(1.75,1.75);
    \draw(1.75,1.75)--(2.25,2.25)[dotted];
    \draw(2.25,2.25)--(4.25,4.25);
    \draw(4.25,4.25)--(4.75,4.75)[dotted];
    \draw(4.75,4.75)--(5,5);
    \draw(5,5)--(5,5.5);
    \draw(5,5)--(5.5,4.5);
    \draw(4,4)--(4.5,3.5);
    \draw(3.5,3.5)--(4,3);
    \draw(3,3)--(3.5,2.5);
    \draw(2.5,2.5)--(3,2);
    \draw(1.5,1.5)--(2,1);
    \draw(1,1)--(1.5,0.5); 
    \node[fill=black,circle,inner sep=1pt, label=above:$\rho$]   at (5,5.5) {};
    \node at (4,5.5) {\large $S$};

    \node[fill=black,circle,inner sep=1pt, label=below:$y_1$]   at (4.5,0.5) {};
    \node[fill=black,circle,inner sep=1pt, label=below:$y_2$]   at (5.5,0.5) {};
    \node[fill=black,circle,inner sep=1pt, label=below:$y_3$]   at (6,1) {};
    \node[fill=black,circle,inner sep=1pt, label=below:$y_m$]   at (7,2) {};
    \node[fill=black,circle,inner sep=1pt, label=below:$x_1$]   at (7.5,2.5) {};
    \node[fill=black,circle,inner sep=1pt, label=below:$x_2$]   at (8,3) {};
    \node[fill=black,circle,inner sep=1pt, label=below:$x_3$]   at (8.5,3.5) {};
    \node[fill=black,circle,inner sep=1pt, label=below:$x_m$]   at (9.5,4.5) {};
    \node[fill=black,circle,inner sep=1pt]   at (5,1) {};
    \node[fill=black,circle,inner sep=1pt]   at (5.5,1.5) {};
    \node[fill=black,circle,inner sep=1pt]   at (6.5,2.5) {};
    \node[fill=black,circle,inner sep=1pt]   at (7,3) {};
    \node[fill=black,circle,inner sep=1pt]   at (7.5,3.5) {};
    \node[fill=black,circle,inner sep=1pt]   at (8,4) {};
    \node[fill=black,circle,inner sep=1pt]   at (9,5) {};
    \draw(4.5,0.5)--(5.75,1.75);
    \draw(5.75,1.75)--(6.25,2.25)[dotted];
    \draw(6.25,2.25)--(8.25,4.25);
    \draw(8.25,4.25)--(8.75,4.75)[dotted];
    \draw(8.75,4.75)--(9,5);
    \draw(9,5)--(9,5.5);
    \draw(9,5)--(9.5,4.5);
    \draw(8,4)--(8.5,3.5);
    \draw(7.5,3.5)--(8,3);
    \draw(7,3)--(7.5,2.5);
    \draw(6.5,2.5)--(7,2);
    \draw(5.5,1.5)--(6,1);
    \draw(5,1)--(5.5,0.5); 
    \node[fill=black,circle,inner sep=1pt, label=above:$\rho$]   at (9,5.5) {};
    \node at (8,5.5) {\large $S'$};

    \node at (13,5.5) {\large $T$};
    \node[fill=black,circle,inner sep=1pt, label=above:$\rho$]   at (14,5.5) {};
    \node[fill=black,circle,inner sep=1pt,  label=below:$x_m^5$]   at (14.5,4.5) {};
    \node[fill=black,circle,inner sep=1pt,  label=below:$x_m^4$]   at (14,4) {};
    \node[fill=black,circle,inner sep=1pt,  label=below:$x_m^3$]   at (13.5,3.5) {};
    \node[fill=black,circle,inner sep=1pt,  label=below:$x_m^2$]   at (13,3) {};
    \node[fill=black,circle,inner sep=1pt,  label=below:$x_m^1$]   at (12.5,2.5) {};
    \node[fill=black,circle,inner sep=1pt,  label=below:$x_2^5$]   at (11.5,1.5) {};
    \node[fill=black,circle,inner sep=1pt,  label=below:$x_2^4$]   at (11,1) {};
    \node[fill=black,circle,inner sep=1pt,  label=below:$x_2^3$]   at (10.5,0.5) {};
    \node[fill=black,circle,inner sep=1pt,  label=below:$x_2^2$]   at (10,0) {};
    \node[fill=black,circle,inner sep=1pt,  label=below:$x_2^1$]   at (9.5,-0.5) {};
    \node[fill=black,circle,inner sep=1pt,  label=below:$x_1^5$]   at (9,-1) {};
    \node[fill=black,circle,inner sep=1pt,  label=below:$x_1^4$]   at (8.5,-1.5) {};
    \node[fill=black,circle,inner sep=1pt,  label=below:$x_1^3$]   at (8,-2) {};
    \node[fill=black,circle,inner sep=1pt,  label=below:$x_1^2$]   at (7.5,-2.5) {};
    \node[fill=black,circle,inner sep=1pt, label=below:$x_1^1$]   at (6.5,-2.5) {};
    \node[fill=black,circle,inner sep=1pt]   at (14,5) {};
    \node[fill=black,circle,inner sep=1pt]   at (13.5,4.5) {};
    \node[fill=black,circle,inner sep=1pt]   at (13,4) {};
    \node[fill=black,circle,inner sep=1pt]   at (12.5,3.5) {};
    \node[fill=black,circle,inner sep=1pt]   at (12,3) {};
    \node[fill=black,circle,inner sep=1pt]   at (11,2) {};
    \node[fill=black,circle,inner sep=1pt]   at (10.5,1.5) {};
    \node[fill=black,circle,inner sep=1pt]   at (10,1) {};
    \node[fill=black,circle,inner sep=1pt]   at (9.5,0.5) {};
    \node[fill=black,circle,inner sep=1pt]   at (9,0) {};
    \node[fill=black,circle,inner sep=1pt]   at (8.5,-0.5) {};
    \node[fill=black,circle,inner sep=1pt]   at (8,-1) {};
    \node[fill=black,circle,inner sep=1pt]   at (7.5,-1.5) {};
    \node[fill=black,circle,inner sep=1pt]   at (7,-2) {};
    \draw(14,5.5)--(14,5);
    \draw(14,5)--(14.5,4.5);
    \draw(14,5)--(11.75,2.75);
    \draw(11.75,2.75)--(11.25,2.25)[dotted];
    \draw(11.25,2.25)--(6.5,-2.5);
    \draw(7,-2)--(7.5,-2.5);
    \draw(7.5,-1.5)--(8,-2);
    \draw(8,-1)--(8.5,-1.5);
    \draw(8.5,-0.5)--(9,-1);
    \draw(9,0)--(9.5,-0.5);
    \draw(9.5,0.5)--(10,0);
    \draw(10,1)--(10.5,0.5);
    \draw(10.5,1.5)--(11,1);
    \draw(11,2)--(11.5,1.5);
    \draw(12,3)--(12.5,2.5);
    \draw(12.5,3.5)--(13,3);
    \draw(13,4)--(13.5,3.5);
    \draw(13.5,4.5)--(14,4);

    \node at (17,5.5) {\large $T'$};
    \node[fill=black,circle,inner sep=1pt, label=above:$\rho$]   at (18,5.5) {};
    \node[fill=black,circle,inner sep=1pt,  label=below:$x_m^1$]   at (18.5,4.5) {};
    \node[fill=black,circle,inner sep=1pt,  label=below:$x_m^4$]   at (18,4) {};
    \node[fill=black,circle,inner sep=1pt,  label=below:$x_m^3$]   at (17.5,3.5) {};
    \node[fill=black,circle,inner sep=1pt,  label=below:$x_m^2$]   at (17,3) {};
    \node[fill=black,circle,inner sep=1pt,  label=below:$x_m^5$]   at (16.5,2.5) {};
    \node[fill=black,circle,inner sep=1pt,  label=below:$x_2^1$]   at (15.5,1.5) {};
    \node[fill=black,circle,inner sep=1pt,  label=below:$x_2^4$]   at (15,1) {};
    \node[fill=black,circle,inner sep=1pt,  label=below:$x_2^3$]   at (14.5,0.5) {};
    \node[fill=black,circle,inner sep=1pt,  label=below:$x_2^2$]   at (14,0) {};
    \node[fill=black,circle,inner sep=1pt,  label=below:$x_2^5$]   at (13.5,-0.5) {};
    \node[fill=black,circle,inner sep=1pt,  label=below:$x_1^1$]   at (13,-1) {};
    \node[fill=black,circle,inner sep=1pt,  label=below:$x_1^4$]   at (12.5,-1.5) {};
    \node[fill=black,circle,inner sep=1pt,  label=below:$x_1^3$]   at (12,-2) {};
    \node[fill=black,circle,inner sep=1pt,  label=below:$x_1^2$]   at (11.5,-2.5) {};
    \node[fill=black,circle,inner sep=1pt, label=below:$x_1^5$]   at (10.5,-2.5) {};
    \node[fill=black,circle,inner sep=1pt]   at (18,5) {};
    \node[fill=black,circle,inner sep=1pt]   at (17.5,4.5) {};
    \node[fill=black,circle,inner sep=1pt]   at (17,4) {};
    \node[fill=black,circle,inner sep=1pt]   at (16.5,3.5) {};
    \node[fill=black,circle,inner sep=1pt]   at (16,3) {};
    \node[fill=black,circle,inner sep=1pt]   at (15,2) {};
    \node[fill=black,circle,inner sep=1pt]   at (14.5,1.5) {};
    \node[fill=black,circle,inner sep=1pt]   at (14,1) {};
    \node[fill=black,circle,inner sep=1pt]   at (13.5,0.5) {};
    \node[fill=black,circle,inner sep=1pt]   at (13,0) {};
    \node[fill=black,circle,inner sep=1pt]   at (12.5,-0.5) {};
    \node[fill=black,circle,inner sep=1pt]   at (12,-1) {};
    \node[fill=black,circle,inner sep=1pt]   at (11.5,-1.5) {};
    \node[fill=black,circle,inner sep=1pt]   at (11,-2) {};
    \draw(18,5.5)--(18,5);
    \draw(18,5)--(18.5,4.5);
    \draw(18,5)--(15.75,2.75);
    \draw(15.75,2.75)--(15.25,2.25)[dotted];
    \draw(15.25,2.25)--(10.5,-2.5);
    \draw(11,-2)--(11.5,-2.5);
    \draw(11.5,-1.5)--(12,-2);
    \draw(12,-1)--(12.5,-1.5);
    \draw(12.5,-0.5)--(13,-1);
    \draw(13,0)--(13.5,-0.5);
    \draw(13.5,0.5)--(14,0);
    \draw(14,1)--(14.5,0.5);
    \draw(14.5,1.5)--(15,1);
    \draw(15,2)--(15.5,1.5);
    \draw(16,3)--(16.5,2.5);
    \draw(16.5,3.5)--(17,3);
    \draw(17,4)--(17.5,3.5);
    \draw(17.5,4.5)--(18,4);
    
    \end{tikzpicture}
    \caption{Left: Two phylogenetic trees $S$ and $S'$ on $2m$ leaves. For $m\geq 3$, it follows that $d^*_\OLA(S,S')<h(S,S')$ since  $h(S,S')=m$ and $d^*_\OLA(S,S')\leq 2$. Right: Two phylogenetic trees $T$ and $T'$ on $5m$ leaves. For $m\ge1$, it follows that $d^*_\PV(T,T')<h(T,T)$ since  $h(T,T')=2m$
    and $d^*_\PV(T,T')\leq m$. For further details, see text. 
}
\label{fig:strict-hybrid}
\end{figure}

To see that the inequality of Theorem~\ref{thm:hybrid-upper-bound} can be strict for trees of arbitrary size under OLA and P2V, Figure~\ref{fig:strict-hybrid} shows two phylogenetic trees $S$ and $S'$ on $2m$ leaves such that, for each $m\geq 3$, $d^*_\OLA(S,S')<h(S,S')$, and also shows two  phylogenetic trees $T$ and $T'$ on $5m$ leaves  such that, for each $m\geq 1$, $d^*_\PV(T,T')<h(T,T)$.  First, consider the ordering $$\sigma(x_1)<\sigma(x_2)<\cdots<\sigma(x_m)<\sigma(y_1)<\sigma(y_2)<\cdots< \sigma(y_m)$$ on the label set of $S$ and $S'$. Then 
\begin{eqnarray*}
\mathbf{v}&=&[0,1,-2,-3,\ldots,-(m-1),-m,-(m+1),-(m+2),-(m+3),\ldots,-(2m-1)],\\
 \mathbf{v}'&=&[0,1,-2,-3,\ldots,-(m-1),1,m+1,-(m+2),-(m+3),\ldots,-(2m-1)],
 \end{eqnarray*}
 where $\mathbf{v}$ and $\mathbf{v}'$ is the OLA vector for $S$ and $S'$, respectively, under $\sigma$ and, thus, $d^*_\OLA(S,S')\leq d_\OLA^\sigma(S,S')=2$. On the other hand, a maximum acyclic agreement forest for $S$ and $S'$ has size $m+1$ and, thus, by Theorem~\ref{thm:hybridnumber}, we have $h(S,S')=m$. Second, consider the ordering $$\sigma(x_1^1)<\sigma(x_1^2)<\sigma(x_1^3)<\sigma(x_1^4)<\sigma(x_1^5)<\cdots<\sigma(x_m^1)<\sigma(x_m^2)<\sigma(x_m^3)<\sigma(x_m^4)<\sigma(x_m^5)$$
 on the label set of $T$ and $T'$.
Then 
\begin{eqnarray*}
\mathbf{w}&=&[0,1,1,5,7,10,12,13,15,17,\ldots,\\
&&10(m-1),10(m-1)+2,10(m-1)+3, 10(m-1)+5,10(m-1)+7],\\
 \mathbf{w}'&=&[0,1,2,5,7,10,11,13,15,17,\ldots,\\
 &&10(m-1),10(m-1)+1,10(m-1)+3,10(m-1)+5,10(m-1)+7],
 \end{eqnarray*}
 where $\mathbf{w}$ and $\mathbf{w}'$ is the P2V vector for $T$ and $T'$, respectively, under $\sigma$ and, thus, $d^*_\PV(T,T')\leq d_\PV^\sigma(T,T')=m$. On the other hand, a maximum acyclic agreement forest for $S$ and $S'$ has size $2m+1$ and, thus, by Theorem~\ref{thm:hybridnumber}, we have $h(T,T')=2m$.

We now establish the equivalence of $d^*_\HOP(T,T')$ and $h(T,T')$ for two phylogenetic trees $T$ and $T'$:

\begin{theorem}\label{thm:hop-equals-hybrid}
Let $T$ and $T'$ be two phylogenetic $X$-trees. Then
$$d^*_{\rm HOP}(T, T') = h(T, T').$$
\end{theorem}

In Theorem~\ref{thm:hybrid-upper-bound}, we already established that the HOP measure $d^*_{\rm HOP}(T,T')$ for two phylogenetic trees $T$ and $T'$ is bounded above by $h(T,T')$. The next lemma shows that $d^*_{\rm HOP}(T,T')$ is also bounded below by $h(T,T')$.

\begin{lemma}\label{lem:inequality2}
Let $T$ and $T'$ be two  phylogenetic $X$-trees, where $|X|=n$. Then
$$d^*_{\rm HOP}(T, T')\ge h(T, T').$$
\end{lemma}

\begin{proof}
Let $\sigma$ be an ordering on $X$.
Consider the following process: {Let $\mathbf{v}_T^\sigma$ and $\mathbf{v}_{T'}^\sigma$ be the HOP vectors for $T$ and $T'$, respectively. Let LCS denote ${\rm LCS}(\mathbf{v}_T^\sigma, \mathbf{v}_{T'}^\sigma$).} 
Set $i=1$. For each element on the path from {the non-leaf vertex $v$ with $f_\HOP(v)=i$} to {the element $x \in X$ with $\sigma(x)=f_\HOP(x)=\ul{i}$} in $T$ that is not in LCS, delete the outgoing edge not on this path. Increment $i$ by $1$ and repeat. Continue this process until $i=n+1$. Let $F=\{T_{\rho}, T_1, T_2, \ldots, T_k\}$ denote the resulting forest. Note that $k=n-|{\rm LCS}|$. Apply this same process to $T'$ to get the forest $F'=\{T'_{\rho}, T'_1, T'_2. \ldots, T'_k\}$.

We complete the proof by showing that $F$ is an acyclic agreement forest for $T$ and $T'$. To this end, we use induction on $n$ to first show that $F$ is isomorphic to $F'$, that is, $F$ is an agreement forest for $T$ and $T'$ before establishing that $F$ is also acyclic. If $n\in \{1, 2\}$, then $T$ is isomorphic to $T'$, and so $F$ is an acyclic agreement forest for $T$ and $T'$. 

Now suppose that the statement holds for all pairs of phylogenetic $X$-trees with $|X| \leq n-1$, where $n\ge 3$. Let $x$ be the element in $X$ with $\sigma(x)=n$, 
and recall that we use $\sigma_{-x}$ to denote the ordering on $X - \{x\}$ such that $\sigma_{-x}(y) = \sigma(y)$ for each element $y \in X - \{x\}$.
Let $T_1 = T|(X - \{x\})$ and $T'_1 = T'|(X - \{x\})$, {and let ${\rm LCS}_1$ denote ${\rm LCS} (\mathbf{v}_{T_1}^{\sigma'}, \mathbf{v}_{T_1'}^{\sigma'}$)}, and $\mathbf{v}_{T_1}^{\sigma'}$ and $\mathbf{v}_{T_1'}^{\sigma'}$ are the respective HOP vectors of $T_1$ and $T_1'$ under $\sigma'=\sigma_{-x}$. Let $F_1=\{S_{\rho}, S_1, S_2, \ldots, S_l\}$ denote the forest obtained from $T_1$ by applying the process that is described in the first paragraph of the proof. By the induction assumption, $F_1$ is an acyclic agreement forest for $T_1$ and $T'_1$.

Before continuing, we make the following observations: (i) ${\rm LCS}_1={\rm LCS}|(X - \{x\})$, where ${\rm LCS}|(X - \{x\})$ refers to the vector obtained from ${\rm LCS}$ by deleting the coordinates corresponding to $x$ (i.e., by deleting the coordinates $\sigma(x)$ and $\ul{\sigma(x)}$), and (ii) the parent of $x$ is labeled by $n$ for each of $T$ and $T'$. Thus, (iii) $F_1$ is isomorphic to $F|(X - \{x\})=\{T_\rho| (X - \{x\}), T_1| (X - \{x\}), \ldots, T_k| (X - \{x\}) \}$ and $F_1$ is isomorphic to $F'|(X- \{x\})=\{T'_\rho| (X - \{x\}), T'_1| (X - \{x\}), \ldots, T'_k| (X - \{x\}) \}$. 

We break the remainder of the proof into two cases depending on whether $n\in {\rm LCS}$. Let $i$ be the element of $\{1,2, \ldots, n-1\}$ such that the subsequence $ [i,\ldots, \underline{i}]$ of $\mathbf{v}_T^\sigma$ contains $n$, and let $i'$ be the element of $\{1,2, \ldots, n-1\}$ such that the subsequence $[i', \ldots, \underline{i'}]$ of $\mathbf{v}_{T'}^\sigma$ contains $n$. To ease reading, we use $j$ to refer to the internal vertex $v$ with $f_\HOP(v)=j$, and $\underline{j}$ to refer to the element $x \in X$ with $\sigma(x)=j$.

First assume that $n\notin {\rm LCS}$. Then the outgoing edge of the internal vertex of $T$ labeled $n$ not on the path from $i$ to $\underline{i}$ is cut. Similarly, the outgoing edge of the internal vertex of $T'$ labeled $n$ not on the path from $i'$ to $\underline{i'}$ is cut. By (ii) and (iii), $F$ is an agreement forest for $T$ and $T'$ and, as $n$ is an isolated vertex, $F$ is an acyclic agreement forest for $T$ and $T'$. 

Second assume that $n\in {\rm LCS}$. Then $i=i'$ and the outgoing edge of the internal vertex labeled $n$ in $T$ (resp.\ $T'$) not on the path from $i$ to $i'$ is not cut. Thus, $\underline n$ is in the same component as $\underline{i}$ in $F$ and $F'$. Let $T_i$ and $T'_i$ denote these components in $F$ and $F'$, respectively. We next show that $T_i$ is isomorphic to $T'_i$, thereby showing, by (iii), that $F$ is an agreement forest for $T$ and $T'$. Keeping the corresponding internal labels of $T$ and $T'$, consider the paths in $T_i$ and $T'_i$ from each of their roots to $\underline{i}$. By construction, the vertex labels on these paths are identical. Furthermore, for each of these vertices, if its label is $j$, then this vertex is the least common ancestor of $\underline j$ and $\underline{i}$ in $T_i$ and $T'_i$. Hence, $T_i$ is isomorphic to $T'_i$, and so $F$ is an agreement forest for $T$ and $T'$. Lastly, we show that $F$ is  acyclic. To this end, let $G$ denote the directed graph associated with $F$ as detailed in the definition of an acyclic agreement forest. For each vertex in $G$, label it with the least leaf value in the corresponding component. Assume that $(i, j)$ is a directed edge in $G$.
Then $i < j$. To see this, we may assume without loss of generality that the root of $T_i$ is an ancestor of the root of $T_j$ in $T$. Since $i$ is the least value in $L(T_i)$, the outgoing edge of the vertex labeled $i$ in $T$ that is on the path from $i$ to $\underline{i}$ is cut. Similarly, the outgoing edge of the vertex labeled $j$ in $T$ that is on the path from $j$ to $\underline{j}$ is cut. But then, for the root of $T_i$ to be an ancestor of the root of $T_j$ in $T$, there must be a path from $i$ to $\underline{j}$ that contains $j$ and each of the these two cut edges. In turn, this implies $i < j$. Returning to $G$, as $\sigma$ is an ordering on $X$, it now follows that $G$ contains no (directed) cycles. Hence $F$ is an acyclic agreement forest for $T$ and $T'$. 
Since $F$ is an acyclic agreement forest for $T$ and $T'$, Theorem~\ref{thm:hybridnumber} implies that $h(T,T') \leq n - |{\rm LCS}|$, and choosing $\sigma$ to be an ordering realizing $d^\ast_{\HOP}(T,T')$ completes the proof.
\end{proof}

Finally, to establish Theorem~\ref{thm:hop-equals-hybrid}, we  combine Theorem~\ref{thm:hybrid-upper-bound} and Lemma~\ref{lem:inequality2}. 
Since computing the hybrid number for two phylogenetic trees is NP-hard~\cite{bordewich2007computing}, the next corollary is an immediate consequence of Theorem~\ref{thm:hop-equals-hybrid} and answers the open problem in~\cite[page 10, Remark 2]{hop}.

\begin{cor} 
Let $T$ and $T'$ be two phylogenetic $X$-trees. Computing $d^*_\HOP(T,T')$ is {\em NP}-hard.
\end{cor}

Let $T$ and $T'$ be two phylogenetic $X$-trees, and let $\sigma$ be an ordering on $X$. In the language of this paper, Zhang et al.~\cite{zhang2023fast} used certain shortest common supersequences of the HOP vectors of $T$ and $T'$ under $\sigma$ to construct a tree-child network that embeds $T$ and $T'$. Furthermore, in Proposition 3 of the supplementary material of~\cite{zhang2023fast}, the authors show that repeating the construction process for each ordering on $X$ can be used to compute $h_{tc}(T,T')$. Hence, since $h(T,T')=h_{tc}(T,T')$ by Lemma~\ref{l:tc-hybrid}, it follows from Theorem~\ref{thm:hop-equals-hybrid} that computing $d_\HOP^*(T,T')$ is equivalent to the shortest common supersequence approach of~\cite{zhang2023fast}.

\section{Bounding order-dependent measures by cherry-picking sequences}
\label{results:cherry-picking}

In this section, we provide sufficient conditions when all three distances $d_\OLA^\sigma(T,T')$, $d_\PV^\sigma(T,T')$, and $d_\HOP^\sigma(T,T')$ are equal and show a connection to the temporal tree-child hybrid number. Let $T$ and $T'$ be two phylogenetic $X$-trees, and let $S=(x_1,x_2,\ldots,x_n)$ be a sequence of the elements in $X$. We call $S$ a {\em cherry-picking sequence} for $T$ precisely if each $x_i$ with $i\in\{1,2,\ldots,n-1\}$ labels a leaf of a cherry in $T|(X - \{x_1,x_2,\ldots,x_{i-1}\})$. Furthermore, if $S$ is a cherry-picking sequence for $T$ and $T'$, then we say that $S$ is a {\em common} cherry-picking sequence for $T$ and $T'$.

Let $S=(x_1,x_2,\ldots,x_n)$ be a common cherry-picking sequence of two  phylogenetic $X$-trees $T$ and $T'$. We say that $x_i$ with $i\in\{1,2,\ldots,n\}$ {\em agrees} if $i=n$, or $x_i$ is a leaf of a cherry $\{x_i,y\}$ in $T|(X- \{x_1,x_2,\ldots,x_{i-1}\})$ and $\{x_i,y\}$ is also a cherry in $T'|(X- \{x_1,x_2,\ldots,x_{i-1}\})$; otherwise, we say that $x_i$ {\em disagrees}. The {\em weight} of $S$, $wt(S)$, is
$$wt(S) = |\{x_i \in S:i\in\{1,2,\ldots,n-1\}\text{ and } x_i \text{ disagrees}\}|.$$
Lastly, we call $S$ a {\em minimum common cherry-picking sequence for $T$ and $T'$} if $wt(S)$ is minimized over all common cherry-picking sequences for $T$ and $T'$. This minimum number is denoted by $s(T,T')$. The next theorem is established in~\cite[Theorem 2]{humphries2013cherry}, where it is also shown that if $T$ and $T'$ have no common cherry-picking sequence, then there is no temporal tree-child network that displays $T$ and $T'$~\cite[Theorem 1]{humphries2013cherry}.

\begin{theorem}\label{t:CPS}
Let $T$ and $T'$ be two phylogenetic $X$-trees. Then $h_t(T,T')=s(T,T')$.
\end{theorem}

Consider two phylogenetic trees $T$ and $T'$ that have a common cherry-picking sequence (a ``natural ordering'' in the terminology of \cite{penn2023leaping}) $S = (x_1, x_2,\ldots, x_n)$. Throughout this section, we refer to the ordering $\sigma$ with $\sigma(x_i)=n-i+1$ for each $i \in \{1,2,\ldots,n\}$ as the ordering on $X$ that is {\em induced} by $S$. The main result of this section is the following theorem.

\begin{theorem}
    Let $T$ and $T'$ be two phylogenetic $X$-trees, and suppose that $S=(x_1, x_2, \ldots, x_n)$ is a common cherry-picking sequence for $T$ and $T'$. Let $\sigma$ be the ordering on $X$ that is induced by $S$.
    Then, 
    $$
    d_{\OLA}^{\sigma}(T,T') = 
    d_{\HOP}^{\sigma}(T,T') =
    d_{\PV}^{\sigma}(T,T') = wt(S).
    $$
    \label{thm:cherry-pick-equiv}
\end{theorem}

We start with a lemma about OLA and P2V vectors.  The ordering associated with a cherry-picking sequence corresponds to a tree-growing process that, at every step, attaches the new leaf to a pendant edge to create a cherry. As such, the resulting vector consists of labels solely from leaves (and not internal vertices).     

\begin{lemma}\label{l:cherry-ordering}
    Let $T$ be a phylogenetic $X$-tree with $|X|=n$, 
    and let $S = (x_1,x_2, \ldots, x_n)$ be a cherry-picking sequence for $T$. Further, let $\sigma$ be the ordering on $X$ that is induced by $S$.
    Let $\mathbf{u}=[u_1,u_2, \ldots, u_n]$ be the $\OLA$ vector of $T$ under $\sigma$, and let $\mathbf{v}=[v_1,v_2, \ldots, v_n]$ be the $\PV$ vector of $T$ under $\sigma$. Then $\mathbf{u} = \mathbf{v}$, i.e., the $\OLA$ and $\PV$ vectors of $T$ are identical under $\sigma$. 
\end{lemma}

\begin{proof}
    The proof is by induction on $n$. If $n=1$, then $\mathbf{u}=\mathbf{v}=[0]$, and if $n=2$, then $\mathbf{u} = \mathbf{v}= [0,1]$, and the statement immediately holds. 
    
    Suppose that $n \geq 3$ and that the statement holds for all  phylogenetic trees with at most $n-1$ leaves. Let $T$ be a phylogenetic $X$-tree with $|X|=n$ and cherry-picking sequence $S = (x_1, x_2,\ldots, x_n)$. By assumption, $\sigma(x_1)=n$. Further, $x_1$ is part of a cherry, say $\{c,x_1\}$, in $T$. Let $T_1 = T|(X - \{x_1\})$. By the induction assumption, for
    the OLA and P2V vectors, say $\mathbf{u}'=[u'_1, u'_2,\ldots, u'_{n-1}]$ and $\mathbf{v}'=[v'_1,v'_2, \ldots, v'_{n-1}]$, of $T_1$, we have $\mathbf{u}' = \mathbf{v}'$. 
    Since $\{c,x_1\}$ is a cherry in $T$, the OLA and P2V vectors, say $\mathbf{u}$ and $\mathbf{v}$, for $T$ under $\sigma$ can be obtained from $\mathbf{u}'$ and $\mathbf{v}'$ by setting $\mathbf{u} = [u_1', u_2',\ldots, u_{n-1}',\sigma(c)]$ and $\mathbf{v} = [v_1',v_2', \ldots, v_{n-1}',\sigma(c)]$, where $\sigma(c)$ is the rank of $c$ under $\sigma$. Since by the induction assumption $u_i' = v_i'$ for each $i \in \{1,2, \ldots, n-1\}$, clearly $\mathbf{u} = \mathbf{v}$, which completes the proof. 
\end{proof}

While the coordinates in the OLA and P2V vectors of a tree disagree when a leaf is attached to an internal edge in the tree-generating process, the coordinates coincide when a leaf is attached to a pendant edge. Hence a consequence of Lemma~\ref{l:cherry-ordering} is the next corollary.

\begin{cor}
    Let $T$ and $T'$ be two  phylogenetic $X$-trees with $|X|=n$. Suppose that $S=(x_1, x_2, \ldots, x_n)$ is a common cherry-picking sequence for $T$ and $T'$. Then there exists an ordering $\sigma$ on $X$ such that
    $$
    d_{\OLA}^{\sigma}(T,T') = 
    d_{\PV}^{\sigma}(T,T').
    $$
    \label{cor:ola-hop-equiv}
\end{cor}

We next show that if $S$ is a common cherry-picking sequence of two phylogenetic $X$-trees $T$ and $T'$, and $\sigma$ is the ordering on $X$ induced by $S$, then the OLA and HOP distances between $T$ and $T'$ under $\sigma$ are the same and equal to $wt(S)$.

\begin{lemma}\label{l:need-a-label}
    Let $T$ and $T'$ be two phylogenetic $X$-trees with $|X|=n$. Suppose that $S=(x_1, x_2, \ldots, x_n)$ is a common cherry-picking sequence for $T$ and $T'$ of weight $wt(S)$. Let $\sigma$ be the ordering on $X$ that is induced by $S$.
    Then
    $$
        d_{\OLA}^{\sigma}(T,T') = 
        d_{\HOP}^{\sigma}(T,T') = wt(S).
    $$
    
    \label{lem:p2v-hop-equiv}
\end{lemma}

\begin{proof}
The proof is by induction on $n$. If $n \in \{1,2\}$, then $T$ and $T'$ are isomorphic, so $d_{\OLA}^{\sigma}(T,T') = d_{\HOP}^{\sigma}(T,T') = 0$. Further, $wt(S)=0$. 

Now assume that $n > 2$ and that the result holds for all pairs of phylogenetic $X$-trees with a common cherry-picking sequence that have at most $n-1$ leaves. 

Let $T$ and $T'$ be two phylogenetic $X$-trees, where $|X|=n$, that have a common cherry-picking sequence $S = (x_1,x_2, \ldots, x_n)$ and let $\sigma$ be the ordering on $X$ that is induced by $S$.
By assumption, $\sigma(x_1)=n$. Let $T_1 = T|(X- \{x_1\})$, and let $T_1' = T'|(X - \{x_1\})$. Let $\mathbf{u}_1=[u_1,u_2,\ldots, u_{n-1}]$ and $\mathbf{u}'_1=[u'_1,u'_2, \ldots, u'_{n-1}]$ be the OLA vectors for $T_1$ and $T_1'$ under $\sigma_{- x_1}$, and let $\mathbf{v}^1=[1,\mathbf{v}^1_1, \underline{1}, \mathbf{v}^1_2, \underline{v_2},\ldots, \mathbf{v}^1_{n-2}, \underline{v_{n-2}},\underline{v_{n-1}}]$ and $\mathbf{\widetilde{v}}^1=[1,\mathbf{\widetilde{v}}^1_1, \underline{1}, \mathbf{\widetilde{v}}^1_2, \underline{\widetilde{v}_2},\ldots, \mathbf{\widetilde{v}}^1_{n-2}, \underline{\widetilde{v}_{n-2}},\underline{\widetilde{v}_{n-1}}]$ be the HOP vectors for $T_1$ and $T_1'$ under $\sigma_{- x_1}$. Let $S_1 = (x_2,x_3, \ldots, x_n)$ be the common cherry-picking sequence for $T_1$ and $T_1'$ obtained from $S$ by omitting $x_1$. 
By the induction assumption, 
\begin{equation} \label{OLA-HOP-S-IA}
    d_{\OLA}^{\sigma_{-x_1}}(T_1,T'_1) = d_{\HOP}^{\sigma_{- x_1}}(T_1,T'_1) = wt(S_1). 
\end{equation}

Now, by assumption $x_1$ is in a cherry, $\{c,x_1\}$ say, in $T$, and in a cherry, $\{c', x_1\}$ say, in $T'$. As in the proof of Lemma~\ref{l:cherry-ordering}, this implies that the OLA vectors, say $\mathbf{u}$ and $\mathbf{u}'$ for $T$ and $T'$ under $\sigma$ can be obtained from $\mathbf{u}_1$ and $\mathbf{u}_1'$ by setting $\mathbf{u} = [u_1, u_2,\ldots, u_{n-1},\sigma(c)]$ and $\mathbf{u}' = [u_1',u_2',\ldots, u'_{n-1},\sigma(c')]$. Further, the HOP labeling of $T$ and $T'$ is obtained from the HOP labeling for $T_1$ and $T_1'$ by assigning label $n$ to the vertex that is introduced when adjoining leaf $x_1$ to $T_1$ (resp. $T_1'$), while keeping all other labels unchanged. Referring to Algorithm~\ref{alg:hop} (Lines~15--17) this implies that for $T$, the vertex labeled $n$ is the last element in $\mathbf{v}(P_{\sigma(c)})$, and for $T'$, it is the last element in $\mathbf{v}(P_{\sigma(c')})$. Thus, the HOP vectors, say $\mathbf{v}$ and $\mathbf{v}'$, for $T$ and $T'$ under $\sigma$, can be obtained from $\mathbf{v}^1$ and $\mathbf{\widetilde{v}}^1$ by inserting (the first occurrence of) element $n$ immediately before the element $\sigma(c)$ (resp. $\sigma(c')$) and appending (the second occurrence of) element $n$ at the end of the resulting vectors.

We now distinguish two cases:
\begin{enumerate}[label={\upshape(\roman*)}]
    \item If $c = c'$, $\{c, x_1\} = \{c',x_1\}$ is a common cherry of $T$ and $T'$, and thus, $wt(S)=wt(S_1)$. Moreover, $d_{\OLA}^{\sigma_{-x_1}}(T_1,T'_1) = d_{\OLA}^\sigma (T,T')$. Furthermore, we have $|\text{LCS}(\mathbf{v}_{\sigma(c)},\mathbf{v}'_{\sigma(c)})| = |\text{LCS}(\mathbf{v}^1_{\sigma_{-x_1}(c)},\mathbf{\widetilde{v}}^1_{\sigma_{- x_1}(c)})| + 1$ and $| \text{LCS}(\mathbf{v}_i, \mathbf{v}'_i)| = | \text{LCS}(\mathbf{v}^1_i, \mathbf{\widetilde{v}}^1_i)|$ for all $i \in \{1, 2,\ldots, n-1\} - \{\sigma(c)\}$. In other words, $\text{Sim}^\sigma_{\HOP}(T,T') = \text{Sim}^{\sigma_{- x_1}}_{\HOP}(T_1,T'_1)+1$. Thus,
    \begin{align*}
        d_{\HOP}^\sigma (T, T') &= n - \text{Sim}^\sigma_{\HOP}(T,T') 
        = n - 1 - \text{Sim}^{\sigma_{- x_1}}_{\HOP}(T_1,T'_1) = d_{\HOP}^{\sigma_{- x_1}}(T_1, T_1'). \\
    \end{align*}
    The statement now follows immediately from Equation~\eqref{OLA-HOP-S-IA}.

    \item If $c \neq c'$, then $\{c,x_1\}$ is a cherry in $T$, while $\{c',x_1\}$ is a cherry in $T'$. For the weights of the common cherry-picking sequence $S$ of $T$ and $T'$ (resp. $S_1$ of $T_1$ and $T_1'$) this implies that $wt(S) = wt(S_1) + 1$. By construction of the OLA vector (see Line 12 of Algorithm~\ref{alg:ola}), $c \neq c'$ implies the last positions of the vector differ.  We have $d_{\OLA}^\sigma(T,T') = d_{\OLA}^{\sigma_{- x_1}}(T_1, T'_1)+1$. Finally, for HOP, since the (first occurrence of) element $n$ is inserted immediately before $\sigma(c)$ in $\mathbf{v}^1$ to obtain $\mathbf{v}$, while it is inserted immediately before $\sigma(c')$ in $\mathbf{\widetilde{v}}^1$ to obtain $\mathbf{v}'$ (and since $c \neq c'$, we have $\sigma(c) \neq \sigma(c')$), we have $| \text{LCS}(\mathbf{v}_i, \mathbf{v}'_i)| = | \text{LCS}(\mathbf{v}^1_i, \mathbf{\widetilde{v}}^1_i)|$ for all $i \in \{1,2, \ldots, n-1\}$, i.e., no LCS changes its length. This implies that $\text{Sim}^\sigma_{\HOP}(T,T') = \text{Sim}^{\sigma_{- x_1}}_{\HOP}(T_1,T'_1)$ and 
    \begin{align*}
        d_{\HOP}^\sigma (T, T') &= n - \text{Sim}^\sigma_{\HOP}(T,T') 
        = n - \text{Sim}^{\sigma_{- x_1}}_{\HOP}(T_1,T'_1) 
        = n - 1 - \text{Sim}^{\sigma_{- x_1}}_{\HOP}(T_1,T'_1) + 1 \\
        &= d_{\HOP}^{\sigma_{- x_1}}(T_1, T_1') + 1. 
    \end{align*}
    The statement now again immediately follows from Equation~\eqref{OLA-HOP-S-IA}. This completes the proof.
\end{enumerate}
 \end{proof}

With Lemmas~\ref{l:cherry-ordering} and~\ref{l:need-a-label}, we can now establish Theorem~\ref{thm:cherry-pick-equiv}:

\begin{proof}[Proof of Theorem~\ref{thm:cherry-pick-equiv}] 
The first equality follows immediately from Corollary~\ref{cor:ola-hop-equiv}, and the second and third equality from Lemma~\ref{lem:p2v-hop-equiv}:
    $$
    d_{\PV}^{\sigma}(T,T') = 
    d_{\OLA}^{\sigma}(T,T') = 
    d_{\HOP}^{\sigma}(T,T') = wt(S).
    $$    
\end{proof}

Let $T$ and $T'$ be two phylogenetic $X$-trees that have a common cherry-picking sequence. If $\sigma$ is an ordering on $X$ that is induced by a cherry-picking sequence that is common to $T$ and $T'$, then we refer to $\sigma$ as a {\em cherry-picking sequence (CPS) ordering} on $X$. Now for each $\Theta\in\{\HOP,\OLA,\PV\}$, define
    $$
        d_\Theta^{CPS}(T,T') = \min_{\sigma} d_\Theta^{\sigma}(T,T'), 
    $$
where the minimum is taken over all CPS orderings $\sigma$ on $X$.
The next corollary follows from Theorems~\ref{t:CPS} and~\ref{thm:cherry-pick-equiv}.

\begin{cor}
Let $T$ and $T'$ be two phylogenetic $X$-trees with a common cherry-picking sequence. Then 
    $$
    d_{\OLA}^{CPS}(T,T') = 
    d_{\HOP}^{CPS}(T,T') =
    d_{\PV}^{CPS}(T,T') = h_t(T,T').
    $$
    \label{temporal-tree-child-hybrid}
\end{cor}

It is worth nothing that for two phylogenetic trees $T$ and $T'$ and any $\Theta\in\{\HOP,\OLA,\PV\}$, the two measures $d_\Theta^*(T,T')$ and $d_\Theta^{CPS}(T,T')$ are not necessarily equal. To see this, consider the two trees that are shown in Figure~\ref{fig:CPS}. Their only two common cherry-picking sequences are $S=(x_3,x_4,x_5,x_1,x_2)$ and $S'=(x_3,x_4,x_5,x_2,x_1)$ with $wt(S)=wt(S')=3$. However, the ordering $\sigma$ on $X$ with $\sigma(x_3)<\sigma(x_4)<\sigma(x_5)<\sigma(x_1)<\sigma(x_2)$, is not a CPS ordering on $X$. It is straightforward to check that $d_\Theta^*(T,T')\leq d_\Theta^\sigma(T,T')=2$.

\begin{figure}[t!]
\begin{center}
    \begin{tikzpicture}[scale=.8]

     \node[fill=black,circle,inner sep=1pt, label=below:$x_1$]   at (0.5,0.5) {};
     \node[fill=black,circle,inner sep=1pt, label=below:$x_3$]   at (1.5,0.5) {};
     \node[fill=black,circle,inner sep=1pt, label=below:$x_4$]   at (2,1) {};
     \node[fill=black,circle,inner sep=1pt, label=below:$x_5$]   at (2.5,1.5) {};
     \node[fill=black,circle,inner sep=1pt, label=below:$x_2$]   at (3,2) {};
     \node[fill=black,circle,inner sep=1pt]   at (1,1) {};
     \node[fill=black,circle,inner sep=1pt]   at (1.5,1.5) {};
     \node[fill=black,circle,inner sep=1pt]   at (2,2) {};
     \node[fill=black,circle,inner sep=1pt]   at (2.5,2.5) {};
     \node[fill=black,circle,inner sep=1pt, label=above:$\rho$]   at (2.5,3) {};
     \draw(2.5,2.5)--(2.5,3);
     \draw(2.5,2.5)--(0.5,0.5);
     \draw(2.5,2.5)--(3,2);
     \draw(2,2)--(2.5,1.5);
     \draw(1.5,1.5)--(2,1);
     \draw(1,1)--(1.5,0.5);
     \node at (1,3) {\large $T$};

     \node[fill=black,circle,inner sep=1pt, label=below:$x_2$]   at (4.5,0.5) {};
     \node[fill=black,circle,inner sep=1pt, label=below:$x_3$]   at (5.5,0.5) {};
     \node[fill=black,circle,inner sep=1pt, label=below:$x_4$]   at (6,1) {};
     \node[fill=black,circle,inner sep=1pt, label=below:$x_5$]   at (6.5,1.5) {};
     \node[fill=black,circle,inner sep=1pt, label=below:$x_1$]   at (7,2) {};
     \node[fill=black,circle,inner sep=1pt]   at (5,1) {};
     \node[fill=black,circle,inner sep=1pt]   at (5.5,1.5) {};
     \node[fill=black,circle,inner sep=1pt]   at (6,2) {};
     \node[fill=black,circle,inner sep=1pt]   at (6.5,2.5) {};
     \node[fill=black,circle,inner sep=1pt, label=above:$\rho$]   at (6.5,3) {};
     \draw(6.5,2.5)--(6.5,3);
     \draw(6.5,2.5)--(4.5,0.5);
     \draw(6.5,2.5)--(7,2);
     \draw(6,2)--(6.5,1.5);
     \draw(5.5,1.5)--(6,1);
     \draw(5,1)--(5.5,0.5);
     \node at (5,3) {\large $T'$};

    \end{tikzpicture}
\end{center}
\caption{Two phylogenetic $X$-trees $T$ and $T'$ with precisely two common cherry-picking sequences, $S = (x_3, x_4, x_5, x_1, x_2)$ and $S' = (x_3, x_4, x_5, x_2, x_1)$.}
\label{fig:CPS}
\end{figure}
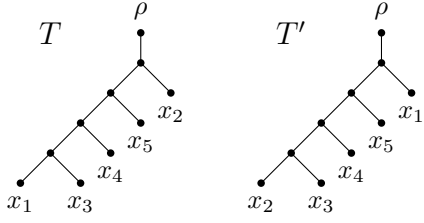

\section{Conclusion and open problems}
\label{openProblems}

Our paper explores the relationship between three novel dissimilarity measures on phylogenetic trees (each of which, given a fixed ordering of the leaf set, takes polynomial time to compute) and the popular, but computationally hard, rSPR distance.  While there is little relationship between these measures and rSPR, we show that direct relationships exist when we compare these measures with the hybrid number and temporal hybrid number. We end the paper with some open problems concerning these intriguing classes of measure.

\subsection*{Capturing the rSPR distance}
The first problem is to develop an order-dependent measure that is efficient to compute but, when minimized over all orders, is equivalent to the rSPR distance. For the order-dependent measures developed thus far, the ordering of the leaves introduces additional structure that allows for polynomial running time to compute the distances.  For example, the HOP ordering provides a clever way to decompose the tree into paths, and for each leaf, the path between it and the root can be viewed as a backbone where leaves later in the ordering are attached. This decomposition captures the agreement forests induced by rSPR moves, albeit well-behaved rSPR moves that do not move a component with a lower minimal element to one with a higher minimal element.  Theorem~\ref{thm:hop-equals-hybrid} shows that this decomposition precludes cycles in the agreement forests and equates HOP to the hybrid number.  Is it possible to keep the efficient running time that the order provides but capture more general rSPR moves?

\begin{problem}  Is there an order-dependent measure $\varphi$ that captures rSPR? That is, for two phylogenetic $X$-trees $T$ and $T'$, if we minimize across all orderings of $X$
$$
d_{\varphi}^*(T,T') = d_\rSPR(T,T').
$$
but, for fixed ordering $\sigma$ on $X$, the value $d^\sigma_{\varphi}(T, T')$ can be computed in polynomial time.
\end{problem}

\subsection*{Sharpening the upper bound for OLA}  
In Theorem~\ref{thm:OLA-20SPR} we proved that for all phylogenetic $X$-trees $T$ and $T'$, there exists an ordering $\sigma$ on $X$ for which $d^\sigma_\OLA(T,T')$ is bounded from above by $28 \cdot d_\rSPR(T,T')$.
Can this upper bound be lowered?  

\begin{problem} Does there exist a constant $c<28$ such that for any two phylogenetic $X$-trees $T$ and $T'$, there exists an ordering $\sigma$ such that
$$
d_\OLA^\sigma(T,T') \leq c\cdot d_\rSPR(T,T')?
$$
\end{problem}

\subsection*{Relating the three order-dependent measures}
In this paper, we have focused on relating the three order-dependent measures OLA, P2V, and HOP to well-established concepts for phylogenetic $X$-trees, including the rSPR distance, the hybrid number, and cherry-picking sequences. However, an interesting future direction is to relate the three order-dependent measures to each other. For instance, one could ask how the minimal measures $d^\ast_\OLA(T,T'), d^\ast_\PV(T,T')$, and $d^\ast_\HOP(T,T')$ are related for two given phylogenetic $X$-trees $T$ and $T'$. Is one always smaller than the other? What can be said about the orderings $\sigma$ that induce the minimum? Are they the same for $d^\ast_\OLA(T,T'), d^\ast_\PV(T,T')$, and $d^\ast_\HOP(T,T')$?

To give a flavor of these types of questions, we consider the relationship between the OLA and HOP measures. In the following, let $T$ be a phylogenetic $X$-tree with at least two leaves. We refer to $T$ as a \emph{caterpillar} if we can order $X$, say $x_1,x_2,\ldots,x_n$, such that  $x_1$ and $x_2$ have the same parent and, for each $i\in\{3,4,\ldots,n\}$, the parent of $x_i$ is the parent of $x_{i-1}$. If $T$ is a caterpillar, then we denote it by $(x_1,x_2,x_3,\ldots,x_n)$ or, equivalently,  $(x_2,x_1,x_3,\ldots,x_n)$. As an example, Figure~\ref{fig:CPS} shows the two caterpillars $T = (x_1, x_3, x_4, x_5, x_2)$ and $T' = (x_2, x_3, x_4, x_5, x_1)$.

First consider the two caterpillars 
$$S = (x_1, x_2, \ldots, x_m, y_1, y_2, \ldots, y_m)
\mbox{ and } 
S' = (y_1, y_2, \ldots, y_{m}, x_1, x_2, \ldots, x_{m})$$ on $2m$ leaves as depicted in Figure~\ref{fig:strict-hybrid}. 
Here, we have $d^*_\OLA(S,S') = d_{\rSPR}(S,S') = 2$ and $d^*_\HOP(S,S') = m$. For OLA, an ordering $\sigma$ on $X$ that realizes $d^*_\OLA(S,S')$ is, for instance, given by
\[ \sigma(x_1) < \sigma(x_2) < \cdots < \sigma(x_{m}) < \sigma(y_1) < \sigma(y_2) < \cdots < \sigma(y_{m}).\]
The same ordering $\sigma$ also realizes $d^*_\HOP(S,S')$.
In summary, the ordering $\sigma$ given above realizes both  $d^*_\OLA(S,S')$ and $ d^*_\HOP(S,S')$. Furthermore, we have $d^*_\OLA(S,S') < d^*_\HOP(S,S')$. In particular, $d^*_\OLA(S,S') = 2$ is a constant, whereas $d^*_\HOP(S,S') = m$ depends on $m$.

Next, consider the two caterpillars 
$$T = (x_1, x_2, y_1, x_3, y_2, \ldots, y_{m-3}, x_{m-1}, y_{m-2}, x_m, y_{m-1})$$
and
$$T' = (x_1, x_m, y_1, x_{m-1}, y_2, \ldots, y_{m-3}, x_3, y_{m-2}, x_2, y_{m-1})$$ on $2m-1$ leaves. Here, we have $d^*_\HOP(T,T')=h(T,T') = m-1$.
An ordering $\sigma$ on $X$ realizing $d^*_\HOP(T,T')$ is  given by \[ \sigma(x_1) < \sigma(x_2) < \cdots < \sigma(x_m) < \sigma(y_1) < \cdots < \sigma(y_{m-1}). \]
Under this ordering $\sigma$, we have $d^\sigma_\OLA(T,T') = 2m-3 > h(T,T')$. By Theorem~\ref{thm:hybrid-upper-bound}, $\sigma$ clearly does not realize $d^\ast_\OLA(T,T')$, indicating that an ordering that realizes the minimum for one of the three-order dependent measures (here, $d^*_\HOP(T,T')$) does not necessarily realize the minimum for the others (here, $d^\ast_\OLA(T,T')$).
 
We thus end by posing the following broad open problem.
\begin{problem} 
Given two rooted phylogenetic $X$-trees $T$ and $T'$ and an ordering $\sigma$ on $X$, what can be said about the relationship of $d^\sigma_\OLA(T,T'), d^\sigma_\PV(T,T')$, and $d^\sigma_\HOP(T,T')$? Further, what can be said about the relationship of $d^\ast_\OLA(T,T'), d^\ast_\PV(T,T')$, and $d^\ast_\HOP(T,T')$ and the orderings realizing them?
\end{problem}

\section{Acknowledgments}

This material is based upon work supported by the National Science Foundation under Grant No.~DMS-1929284 while the authors were in residence at the Institute for Computational and Experimental Research in Mathematics in Providence, RI, during the Theory, Methods, and Applications of Quantitative Phylogenomics semester program. The first and third authors thank the New Zealand Marsden Fund for their financial support.

\bibliographystyle{acm}
\bibliography{refs}

\end{document}